\documentclass[11pt,a4paper]{article}

\usepackage[margin=1in]{geometry}
\usepackage[T1]{fontenc}
\usepackage{graphicx}
\usepackage[square,numbers,sort&compress]{natbib}
\usepackage{amsmath,amssymb,amsfonts,amsthm}
\usepackage{algorithm,algorithmic}
\usepackage{placeins}
\usepackage{bm}
\usepackage{dsfont}
\usepackage{bbold}
\usepackage{xcolor}
\usepackage{url}
\usepackage{microtype}
\usepackage[hidelinks,bookmarksnumbered]{hyperref}

\hypersetup{
  pdftitle={Accelerated Push-Pull with Local Updates for Decentralized Stochastic Optimization},
  pdfauthor={Sheng Dong, Weiguo Xia, Huijie Tao},
  pdfsubject={Decentralized stochastic nonconvex optimization over directed networks},
  pdfkeywords={Decentralized optimization, stochastic optimization, directed network, momentum method, local updates}
}

\theoremstyle{plain}
\newtheorem{theorem}{Theorem}
\newtheorem{lemma}{Lemma}
\newtheorem{corollary}{Corollary}
\newtheorem{assumption}{Assumption}
\theoremstyle{remark}
\newtheorem{remark}{Remark}

\renewenvironment{proof}[1][Proof]
  {\par\noindent\textit{#1.}\ \ignorespaces}
  {\par\medskip}

\begin{document}

\title{Accelerated Push--Pull with Local Updates for Decentralized Stochastic Optimization\thanks{This work was supported by the National Natural Science Foundation of China (92567303). The material in this paper was partially presented at the 23rd World Congress of the International Federation of Automatic Control (IFAC 2026), August 23--28, 2026, Busan, Republic of Korea~\citep{Dong2026}.}}
\author{
  Sheng Dong \quad Weiguo Xia\thanks{Corresponding author.} \quad Huijie Tao\\[0.6em]
  \parbox{0.95\textwidth}{\centering\small
    Key Laboratory of Intelligent Control and Optimization for Industrial Equipment\\
    of Ministry of Education and School of Control Science and Engineering,\\
    Dalian University of Technology, Dalian, China\\[0.4em]
    \texttt{dongsheng@mail.dlut.edu.cn},  \texttt{wgxiaseu@dlut.edu.cn},\\\texttt{taohuijie7yx@mail.dlut.edu.cn}
  }
}
\date{}

\maketitle

\begin{abstract}
We consider decentralized stochastic nonconvex optimization over directed networks. Push--Pull provides an efficient decentralized optimization framework for such networks, but in the stochastic setting, gradient variance can slow its convergence and complicate the analysis of linear speedup. Motivated by this observation, we incorporate recursive momentum into stochastic Push--Pull for variance reduction, yielding the accelerated stochastic Push--Pull (ASPP) algorithm. Under mean-squared smoothness, we prove that ASPP improves the convergence rate of stochastic Push--Pull and achieves asymptotic linear speedup with respect to the number of agents. To further reduce communication, we develop LASPP, a local-update variant of ASPP that allows each agent to perform multiple local updates between communication rounds. We show that LASPP achieves an additional asymptotic linear speedup with respect to the number of local steps. With a suitable number of local steps, LASPP reduces communication complexity while maintaining the same stochastic gradient complexity as ASPP. Numerical experiments on real datasets demonstrate the effectiveness of the proposed algorithms.
\end{abstract}


\section{Introduction}

The increasing computational demands of machine learning and the distributed storage of training data motivate the development of distributed optimization methods. We consider $n$ agents, indexed by $\mathcal N:=\{1,\ldots,n\}$, that jointly solve the following stochastic optimization problem
\begin{equation}\label{eq1}
	\min_{x\in\mathbb R^p} f(x):=\frac1n\sum_{i=1}^{n}f_i(x).
\end{equation}
Here, $f_i(x):=\mathbb E_{\xi_i\sim\mathcal D_i}\bigl[F_i(x;\xi_i)\bigr]$ is the local objective function of agent~$i$, and $\xi_i$ is a sample drawn from its local data distribution $\mathcal D_i$. We allow nonconvex local objectives and heterogeneous data distributions across agents. This formulation covers distributed empirical risk minimization, where $x$ represents the model parameters and $F_i(x;\xi_i)$ is the sample loss. If agent~$i$ holds a dataset $\mathcal S_i=\{\xi_{i,s}\}_{s=1}^{N_i}$, taking $\mathcal D_i$ to be the uniform distribution over $\mathcal S_i$ yields $f_i(x)=\frac{1}{N_i}\sum_{s=1}^{N_i}F_i(x;\xi_{i,s})$. 

We focus on decentralized methods, in which each agent computes stochastic gradients from its local data and exchanges information with its neighbors over a communication network. These methods avoid the communication bottleneck and single point of failure associated with a central server \citep{lian2017can,assran2019stochastic}. To accommodate asymmetric information exchange among agents, we consider directed communication networks. Suitably designed directed topologies can also combine sparse connectivity with fast information mixing \citep{song2022communication,you2024b,liang2025understanding}.

Communication overhead can, however, limit the benefits of parallel computation in decentralized training. Even with sparse connectivity, frequent exchanges of high-dimensional model parameters may consume substantial training time \citep{wang2025promise}. For methods that communicate at every iteration, faster convergence can reduce the number of communication rounds needed to reach a prescribed accuracy \citep{huang2025accelerated}. Moreover, when communication is costly relative to local computation, performing multiple local updates between communication rounds offers another natural way to reduce communication overhead \citep{fedavg,mishchenko2022proxskip}. These considerations motivate us to investigate how to accelerate decentralized optimization while reducing communication through local updates over directed networks.

\subsection{Related Work}

A central difficulty in decentralized optimization over directed graphs lies in constructing doubly stochastic mixing matrices. To address this difficulty, Push-Sum-based methods employ column-stochastic mixing \citep{nedic2014distributed}, with Push-DIGing further incorporating gradient tracking into this framework \citep{nedic2017achieving}. Another line of work relies only on row-stochastic mixing \citep{mai2019distributed,xi2018linear}. Both types of methods use auxiliary iterations to compensate for the imbalance induced by directed mixing. The Push--Pull method \citep{pu2020push}, also known as the AB algorithm \citep{xin2018linear}, instead uses a row-stochastic matrix to mix decision variables and a column-stochastic matrix for gradient tracking. This combination avoids auxiliary Perron eigenvector estimation. In deterministic settings, Push--Pull has been shown to converge linearly for smooth strongly convex objectives with a sufficiently small step size. It has also exhibited faster empirical convergence than methods relying solely on row- or column-stochastic mixing \citep{pu2020push,xin2018linear,Nedich03092025}.

Recent studies have investigated stochastic Push--Pull (SPP) and demonstrated its empirical advantages over other decentralized stochastic methods \citep{you2025stochastic,liang2025linear}. Despite this empirical success, the theoretical analysis of SPP is challenging. In particular, a key issue is establishing linear speedup with respect to the number of agents, a property that characterizes the convergence benefit of parallel computation. Since SPP uses row- and column-stochastic mixing simultaneously, a direct application of standard single-step descent analysis can yield an overly large variance contribution in the convergence bound, making linear speedup difficult to establish \citep{liang2025linear}. Recent work addresses this difficulty through a multi-step descent analysis and establishes linear speedup for SPP over strongly connected directed graphs \citep{liang2025linear}. Linear speedup has also been established for Spanning Tree Push--Pull, but its analysis relies on specially constructed spanning-tree mixing matrices rather than general row- and column-stochastic matrices \citep{you2026distributed}. Moreover, although SPP uses gradient tracking to mitigate the effect of data heterogeneity, its update directions remain affected by sampling noise. High sampling variance can therefore slow convergence.

The effects of stochastic gradient variance on both the analysis and convergence of SPP motivate the use of variance reduction. Recursive momentum provides such an approach for stochastic nonconvex optimization \citep{cutkosky2019momentum,tran2022hybrid}. In decentralized optimization, GT-HSGD combines recursive momentum with gradient tracking to obtain improved stochastic gradient complexity under mean-squared smoothness \citep{xin2021hybrid}. It also achieves asymptotic linear speedup with respect to the number of agents, but these guarantees rely on doubly stochastic mixing. More recently, Push-ASGD combines recursive momentum with Push-Sum-based gradient tracking to accelerate stochastic nonconvex optimization over directed networks without requiring doubly stochastic mixing \citep{chen2024accelerated}. However, its analysis does not establish linear speedup with respect to the number of agents.

The decentralized stochastic methods discussed above require communication at every iteration. When communication is costly relative to local computation, performing multiple local updates between communication rounds is a common strategy for reducing communication overhead \citep{alghunaim2024local,ren2026communication,huang2025distributed}. This strategy has been widely adopted in federated learning~\citep{fedavg,cheng2024momentum,zhang2026non}. However, data heterogeneity across agents can cause local iterates to drift apart, potentially offsetting the benefits of less frequent communication \citep{SCAFFOLD}. Decentralized methods such as K-GT \citep{liu2025decentralized} and ST-GT \citep{huang2025beyond} address local drift through modified gradient-tracking updates, but their analyses rely on doubly stochastic mixing matrices. For directed networks, DFedSGPSM \citep{li2025asymmetrically} uses the Push-Sum protocol to support local updates without requiring doubly stochastic mixing matrices. However, its analysis does not establish linear speedup with respect to the number of local update steps. Such a guarantee would show that additional local computation can reduce the number of communication rounds needed to reach a prescribed accuracy.

\subsection{Contributions and Organization}

In this paper, we develop accelerated and communication-efficient decentralized algorithms to find $\epsilon$-accurate first-order stationary points of \eqref{eq1}. The proposed algorithms build on Push--Pull and operate over directed networks without requiring doubly stochastic mixing matrices. Our main contributions are summarized as follows.

\begin{itemize}
	\item We propose Accelerated Stochastic Push--Pull (ASPP), which incorporates recursive momentum into SPP for variance reduction to improve convergence and establish linear speedup in a unified manner. Under mean-squared smoothness, we establish network-independent asymptotic communication and per-agent stochastic gradient complexities of $\mathcal O(n^{-1}\epsilon^{-3})$. ASPP thus improves upon the $\mathcal O(n^{-1}\epsilon^{-4})$ complexity guarantee for SPP \citep{liang2025linear} and achieves asymptotic linear speedup with respect to the number of agents. To our knowledge, ASPP is the first method to achieve these complexity guarantees for decentralized stochastic nonconvex optimization without requiring doubly stochastic mixing.
	
	\item To further reduce communication overhead, we develop Local Accelerated Stochastic Push--Pull (LASPP), which allows each agent to perform $K$ local updates between communication rounds. We establish an asymptotic communication complexity of $\mathcal O(n^{-1}K^{-1}\epsilon^{-3})$, which demonstrates linear speedup with respect to both the number of local steps and the number of agents. With a suitable choice of $K$, LASPP improves the asymptotic communication complexity from $\mathcal O(\epsilon^{-3})$ for ASPP to $\mathcal O(\epsilon^{-2})$ without increasing the per-agent stochastic gradient complexity. To our knowledge, these complexity guarantees have not previously been established for decentralized stochastic nonconvex optimization with local updates over directed networks.
\end{itemize}

These results suggest that variance reduction can benefit decentralized stochastic optimization beyond accelerating convergence. By controlling network-dependent variance contributions, recursive momentum can also help reveal convergence properties that are otherwise difficult to establish over directed networks. With suitably designed local updates, these benefits can be retained with less frequent communication.

This paper extends our preliminary conference version \citep{Dong2026} in four main respects. First, we refine the convergence analysis of ASPP and establish network-independent asymptotic complexity guarantees. Second, we propose LASPP, a local-update variant of ASPP designed to reduce communication overhead, and establish its convergence guarantees. Third, we present additional numerical experiments to demonstrate the effectiveness of the proposed algorithms and support our theoretical findings. Finally, we provide detailed proofs that were omitted from the conference version.

The remainder of this paper is organized as follows. Section~2 presents the preliminaries and assumptions. Sections~3 and~4 present ASPP and LASPP, respectively, together with their convergence analyses. Section~5 reports numerical experiments, and Section~6 concludes the paper.

\textit{Notation}: Let $\mathbb R := (-\infty, +\infty)$, $\mathbb N := \{ 0 , 1, 2, \ldots \}$, and $\mathbb N_+ :=\{1,2,\ldots\}$. $\mathbb R^n$ denotes the $n$-dimensional Euclidean space. For $v \in \mathbb R^p$, $\lVert v \rVert$ and $\lVert v \rVert_1$ denote its Euclidean norm and $l_1$-norm, respectively. For $v_1, v_2 \in \mathbb R^n$, $\langle v_1, v_2 \rangle$ denotes their inner product. For $M \in \mathbb R^{n \times p}$, $\lVert M \rVert_F$ and $\lVert M \rVert_2$ denote its Frobenius norm and spectral norm, respectively. Let $\mathds{1}$ denote an all-ones vector and $\mathbb{0}$ denote an all-zero vector or matrix, both of appropriate dimensions. Let the symbol $\lesssim$ denote inequality up to absolute constants.

\section{Preliminaries and Assumptions}

This section specifies the communication model and the assumptions on the local objectives and stochastic gradients for problem~\eqref{eq1}.

The agents communicate over a directed graph $\mathcal G=(\mathcal N,\mathcal E)$, where $\mathcal E\subseteq\mathcal N\times\mathcal N$ is the set of directed edges. An edge $(j,i)\in\mathcal E$ indicates that agent~$j$ can directly send information to agent~$i$. For a nonnegative matrix $M=[m_{ij}]\in\mathbb R^{n\times n}$, let $\mathcal G_M=(\mathcal N,\mathcal E_M)$ denote its induced graph, where $(j,i)\in\mathcal E_M$ if and only if $m_{ij}>0$. The in-neighbor and out-neighbor sets of agent~$i$ in $\mathcal G_M$ are denoted by $\mathcal N_{M,i}^{\mathrm{in}}:=\{j\in\mathcal N\mid(j,i)\in\mathcal E_M\}$ and $\mathcal N_{M,i}^{\mathrm{out}}:=\{j\in\mathcal N\mid(i,j)\in\mathcal E_M\}$, respectively.

For a query point $x\in\mathbb R^p$, let $g_i(x,\xi_i)$ denote the stochastic gradient available to agent~$i$. We make the following assumptions on the local objectives and stochastic gradients~\citep{cutkosky2019momentum,tran2022hybrid,xin2021hybrid,chen2024accelerated}.

\begin{assumption} \label{assumption1}
	For all $i \in \mathcal N$, the local stochastic gradient is unbiased and has bounded variance, i.e., there exists $\sigma \ge 0$ such that for all $x \in \mathbb R^p$,
	$$
	\mathbb E\bigl[g_i(x, \xi_i) \bigr] = \nabla f_i(x) \:\: \text{and} \:\:
	\mathbb E\bigl[ \lVert g_i(x, \xi_i) - \nabla f_i(x) \rVert^2 \bigr] \le \sigma^2.
	$$
    Moreover, all samples drawn by the agents are mutually independent.
\end{assumption}
\begin{assumption} \label{assumption2}
	For all $i \in \mathcal N$, the local loss function $f_i$ is bounded below, i.e., $f_i(x) \ge \inf_{x \in \mathbb R^p} f_i(x) > - \infty$.
\end{assumption}
Given Assumption~\ref{assumption2}, we let $\Delta_f$ be a constant such that $f_i(x^{(0)}) - \inf_{x \in \mathbb R^p} f_i(x) \le \Delta_f$ for all $i \in \mathcal N$, where $x^{(0)}$ is the starting point.
\begin{assumption} \label{assumption3}
	For all $i \in \mathcal N$, the local stochastic gradient satisfies the mean-squared smoothness condition, i.e., 
	there exists $L > 0$ such that for all $x, y \in \mathbb R^p$,
	$$
	\mathbb E \bigl[\lVert g_i (x, \xi_i) - g_i (y, \xi_i) \rVert^2\bigr] \le L^2 \lVert x - y \rVert^2.
	$$
\end{assumption}
Assumptions~\ref{assumption1} and~\ref{assumption3} imply that each $f_i$ is $L$-smooth, i.e., for all $x,y\in\mathbb R^p$, \[ \lVert\nabla f_i(x)-\nabla f_i(y)\rVert\le L\lVert x-y\rVert. \]

\section{Accelerated Stochastic Push--Pull}

In this section, we introduce the Accelerated Stochastic Push--Pull algorithm for decentralized stochastic nonconvex optimization over directed graphs. We first motivate its algorithmic design and present the resulting algorithm, and then establish its convergence guarantees.

\subsection{Algorithm Design}

To motivate our algorithmic design, we first recall the stochastic Push--Pull (SPP) algorithm \citep{liang2025linear}, whose updates are given by
\begin{subequations}\label{eq2}
	\begin{align}
		y_i^{(t+1)}\! &= \!\sum_{j=1}^{n} b_{ij} y_j^{(t)} \!+ g_i(x_i^{(t)}\!, \xi_i^{(t)}) - g_i(x_i^{(t-1)}\!, \xi_i^{(t-1)}), \label{eq2a} \\
		x_i^{(t+1)} \! &= \! \sum_{j=1}^{n} a_{ij} x_j^{(t)} - \alpha y_i^{(t+1)}. \label{eq2b}
	\end{align}
\end{subequations}
Here, $x_i^{(t)}\in\mathbb{R}^p$ and $y_i^{(t)}\in\mathbb{R}^p$ denote the local decision variable and gradient-tracking variable maintained by agent~$i$ at round~$t$, respectively, and $\alpha>0$ is the step size. The mixing matrices $A=[a_{ij}]$ and $B=[b_{ij}]$ are row-stochastic and column-stochastic, respectively. At each communication round~$t$, agent~$i$ pushes $b_{ji}y_i^{(t)}$ to each out-neighbor $j\in\mathcal{N}_{B,i}^{\mathrm{out}}$, thereby implementing column-stochastic mixing in the gradient-tracking update. It also pulls $x_j^{(t)}$ from each in-neighbor $j\in\mathcal{N}_{A,i}^{\mathrm{in}}$ and combines these variables using the weights $a_{ij}$, thereby implementing row-stochastic mixing in the decision-variable update.

The updates in \eqref{eq2} highlight two main advantages of the Push--Pull architecture. First, through gradient tracking, Push--Pull enables each agent to estimate global gradient information, thereby mitigating the effect of data heterogeneity. Second, by jointly using row- and column-stochastic mixing, Push--Pull avoids the auxiliary iterations for Perron eigenvector estimation required by methods relying solely on row- or column-stochastic mixing. Such iterations may slow convergence, particularly when the communication graph is poorly connected \citep{xin2020general}.

Despite these advantages, SPP still faces two main challenges. First, although SPP uses gradient tracking to mitigate the effect of data heterogeneity, its update directions remain affected by sampling noise in local stochastic gradients. This sampling noise can have high variance and consequently slow convergence. Second, since SPP combines row- and column-stochastic mixing, an additional tracking error term appears in the descent bound. Bounding this term in a conventional single-step descent analysis can yield an overly large stochastic gradient variance term, making linear speedup difficult to establish~\citep{liang2025linear}.

A key observation is that overcoming both challenges requires controlling stochastic gradient variance. A natural idea is therefore to employ variance reduction to address them simultaneously. Following this idea, we incorporate a recursive momentum estimator \citep{cutkosky2019momentum,tran2022hybrid,xin2021hybrid} into SPP, yielding the Accelerated Stochastic Push--Pull (ASPP) algorithm. In particular, at iteration $t\geq1$, each agent~$i$ constructs the following recursive momentum estimator
\begin{equation}\label{eq3} 
	m_i^{(t)} = g_i(x_i^{(t)}, \xi_i^{(t)}) + \beta \bigl(m_i^{(t-1)} - g_i(x_i^{(t-1)}, \xi_i^{(t)})\bigr),
\end{equation}
where $\beta\in[0,1)$ is the momentum parameter and both stochastic gradients are computed from the same sample $\xi_i^{(t)}$. Using this estimator in place of the stochastic gradient, agent~$i$ performs the following Push--Pull updates
$$
\begin{aligned}
	y_i^{(t+1)}\! &= \!\sum_{j=1}^{n} b_{ij} y_j^{(t)} + m_i^{(t)} - m_i^{(t-1)}, \\
	x_i^{(t+1)} \! &= \! \sum_{j=1}^{n} a_{ij} x_j^{(t)} - \alpha y_i^{(t+1)}. 
\end{aligned}
$$
The complete procedure over $T$ communication rounds is presented in Algorithm~\ref{alg1}. Here, we initialize all agents at a common point $x^{(0)}\in\mathbb{R}^p$, and each agent computes $m_i^{(0)}$ using an initial mini-batch of size $b_0$.

By using recursive momentum estimates in its tracking update, ASPP mitigates the effect of sampling noise on its update directions, thereby enabling a faster convergence rate. Recursive momentum also allows us to control the variance contribution of the tracking error in a single-step descent analysis. In the next subsection, we show how this variance control leads to a faster convergence rate than SPP and linear speedup with respect to the number of agents.

	\begin{algorithm}
		\caption{Accelerated Stochastic Push--Pull}
		\begin{algorithmic}[1]\label{alg1}
		\normalsize
			\setlength{\itemsep}{0.2em}
			\STATE \textbf{Initialization:} $x_i^{(0)} = x^{(0)}$, $y_i^{(0)} = \mathbb{0}$, $m_i^{(-1)} = \mathbb{0}$
			\STATE \hspace{\algorithmicindent}$m_i^{(0)} = \frac{1}{b_0} \sum_{r=1}^{b_0} g_i(x_i^{(0)}, \xi_{i,r}^{(0)})$
			\STATE \hspace{\algorithmicindent}$y_i^{(1)} = \sum_{j=1}^{n} b_{ij} y_j^{(0)} + m_i^{(0)} - m_i^{(-1)}$
			\STATE \hspace{\algorithmicindent}$x_i^{(1)} = \sum_{j=1}^{n} a_{ij} x_j^{(0)} - \alpha y_i^{(1)}$
			\FOR{ $t = 1, 2, \ldots, T-1$ }
			\setlength{\itemsep}{0.2em}
			\STATE
			$
			m_i^{(t)} \! = g_i(x_i^{(t)}, \xi_i^{(t)}) + \beta \bigl(m_i^{(t-1)} - g_i(x_i^{(t-1)}, \xi_i^{(t)})\bigr)
			$
			\STATE $y_i^{(t+1)} = \sum_{j=1}^{n} b_{ij} y_j^{(t)} + m_i^{(t)} - m_i^{(t-1)}$
			\STATE $x_i^{(t+1)} = \sum_{j=1}^{n} a_{ij} x_j^{(t)} - \alpha y_i^{(t+1)}$
			\ENDFOR
		\end{algorithmic}
	\end{algorithm}

We close this subsection with the following assumption on the mixing matrices $A$ and $B$, which is standard in the analysis of Push--Pull methods \citep{xin2018linear,Nedich03092025,liang2025linear}.
\begin{assumption} \label{assumption4}
	The mixing matrices $A$ and $B$ are row-stochastic and column-stochastic, respectively, and both have positive diagonal entries. Moreover, the directed graphs $\mathcal{G}_A$ and $\mathcal{G}_B$ induced by $A$ and $B$, respectively, are strongly connected.
\end{assumption}
To satisfy Assumption~\ref{assumption4}, we can take $\mathcal G_A=\mathcal G_B=\mathcal G$ when the communication graph $\mathcal G$ is strongly connected and has a self-loop at each node. The weights in $A$ and $B$ can then be constructed locally using in-degree and out-degree information, respectively. The following lemma will be used in the subsequent analysis; see Proposition~1 in~\citep{liang2025linear}.
\begin{lemma} \label{lemma1}
	Under Assumption~\ref{assumption4}, there exist entrywise positive Perron vectors $\pi_A$ and $\pi_B$ such that $\pi_A^{\top} A = \pi_A^{\top}$, $B \pi_B = \pi_B$, $\lVert\pi_A \rVert_1 = \lVert\pi_B \rVert_1 = 1$, and $\pi_A^{\top} \pi_B \in (0,1]$. Moreover, there exist constants $p_A, p_B > 0$ and $\lambda_A, \lambda_B \in [0, 1)$ such that, for all $t \ge 0$,
	\begin{equation} \label{lemmaeq1}
		\lVert A^t - \mathds{1} \pi_A^\top \rVert_2 \le p_A \lambda_A^t, \quad \lVert B^t - \pi_B \mathds{1}^\top \rVert_2 \le p_B \lambda_B^t.
	\end{equation}
\end{lemma}

\subsection{Convergence Analysis}

In this subsection, we establish the convergence guarantees for ASPP. After introducing the necessary preliminaries, we present the analysis leading to accelerated convergence and linear speedup with respect to the number of agents.

\subsubsection{Preliminaries}

We first introduce some notation used throughout the analysis. Let
$$\mathbf{x}^{(t)} := [x^{(t)}_1, x^{(t)}_2, \ldots, x^{(t)}_n]^\top \in \mathbb R^{n \times p},$$
$$\mathbf{y}^{(t)} := [y^{(t)}_1, y^{(t)}_2, \ldots, y^{(t)}_n]^\top \in \mathbb R^{n \times p},$$
$$\mathbf{m}^{(t)} := [m^{(t)}_1, m^{(t)}_2, \ldots, m^{(t)}_n]^\top \in \mathbb R^{n \times p},$$
$${\nabla f}(\mathbf x^{(t)}) := [{\nabla f_1}(x_1^{(t)}), {\nabla f_2}(x_2^{(t)}), \ldots, {\nabla f_n}(x_n^{(t)})]^\top \in \mathbb R^{n \times p}.$$
We also define
$$\hat x^{(t)} := \pi_A^\top {\mathbf x}^{(t)}, \:\:\: {\bar {\mathbf x}}^{(t)} := \mathds{1} {\hat x ^{(t)}}, \:\:\: {\bar {\mathbf y}}^{(t)} := \pi_B \mathds{1}^\top {\mathbf y}^{(t)},$$
$$\bar m^{(t)} := n^{-1} \mathds{1}^\top {\mathbf m}^{(t)}, \:\:\: {\overline {\nabla f}} (\mathbf x^{(t)}) := n^{-1} \mathds{1}^\top {\nabla f}(\mathbf x^{(t)}).$$
Here, $\hat{x}^{(t)}$ denotes the consensus value of the local iterates under the row-stochastic mixing matrix $A$, and thus serves as the desired output of the algorithm at round~$t$.

To quantify the effect of the mixing matrices, we introduce the following parameters~\citep{liang25c,liang2025linear}
$$
s_A := \sum_{\tau=0}^{\infty} \lVert A^\tau - \mathds{1} \pi_A^\top\rVert_2, \quad s_B := \sum_{\tau=0}^{\infty} \lVert B^\tau - \pi_B \mathds{1}^\top\rVert_2, \quad c := n \pi_A^\top \pi_B.
$$
Both $s_A$ and $s_B$ are finite by the exponential decay established in Lemma~\ref{lemma1}. We assume $n\ge2$, which implies $s_A,s_B\ge1$. Proposition~2 in~\citep{liang2025linear} relates $s_A$, $s_B$, and $c$ to other metrics used to characterize mixing over directed graphs. For convenience, we further define $h\ge0$ by
\[h^2:=ns_A^2+ns_B^2+s_A^2s_B^2.\]

For $t\ge0$, the tracking and decision variable updates of ASPP can be written in matrix form as
\begin{align} \mathbf y^{(t+1)}&=B\mathbf y^{(t)}+\mathbf m^{(t)}-\mathbf m^{(t-1)},\label{B.0.1}\\ \mathbf x^{(t+1)}&=A\mathbf x^{(t)}-\alpha\mathbf y^{(t+1)}.\label{B.0.2} \end{align}
Since $B$ is column-stochastic and $\mathbf y^{(0)}=\mathbf m^{(-1)}=\mathbb{0}$, induction on the tracking recursion gives
\begin{equation}\label{B.0.4} \mathds{1}^\top\mathbf y^{(t+1)}=\mathds{1}^\top\mathbf m^{(t)}. \end{equation}
Combining this identity with $\pi_A^\top A=\pi_A^\top$ and the definitions above yields
\begin{equation}\label{B.0.3} \begin{aligned}
		\hat x^{(t+1)}
		&= \hat x^{(t)} -\alpha\pi_A^\top\mathbf y^{(t+1)}\\
		&= \hat x^{(t)} -c\alpha\bar m^{(t)} -\alpha\pi_A^\top\bigl(\mathbf y^{(t+1)}-\bar{\mathbf y}^{(t+1)}\bigr).
	\end{aligned} \end{equation}

Since $f$ is nonconvex, we consider convergence to first-order stationary points. In the stochastic setting, we say that an algorithm finds an $\epsilon$-accurate first-order stationary point after $T$ communication rounds if
\[\frac1T\sum_{t=0}^{T-1}\mathbb E \bigl[\lVert\nabla f(\hat x^{(t)})\rVert^2\bigr]\le\epsilon^2.\]

\subsubsection{Convergence Guarantees}

In this subsection, we present the convergence analysis of ASPP. To derive a descent inequality for the global objective $f$, we examine the recursion for $\hat x^{(t)}$ in \eqref{B.0.3}. The recursion contains an additional tracking error term arising from the simultaneous use of row- and column-stochastic mixing. For SPP, directly bounding this term in a conventional single-step descent analysis can yield an excessively large variance contribution, obscuring linear speedup. This difficulty has been addressed through a multi-step descent analysis \citep{liang2025linear}. Here, we exploit the recursive momentum structure of ASPP to control the variance terms in a single-step descent analysis, thereby establishing linear speedup while also obtaining a faster convergence rate than SPP. 

Using the $L$-smoothness of $f$ and $f_i$, we first obtain the following descent inequality.

\begin{lemma} \label{lemma2}
	Suppose Assumptions~\ref{assumption1}--\ref{assumption4} hold. If $0 <\alpha \le \frac{1}{4cL}$, then the following inequality holds
	\begin{equation} \label{lemmaeq2}
		\begin{aligned}
			&\sum_{t=0}^{T-1} \mathbb E \bigl[ \lVert\nabla f({\hat x}^{(t)}) \rVert^2 \bigr] \le \: \frac{4 \Delta_f}{c \alpha} + \frac{4L^2}{n} \sum_{t=0}^{T-1} \mathbb E \bigl[ \lVert {\mathbf{x}}^{(t+1)} - \bar {\mathbf{x}}^{(t+1)} \rVert_F^2\bigr]\\
			&\quad + 4 \sum_{t=0}^{T-1} \mathbb E \bigl[ \lVert {\bar m}^{(t)} - {\overline {\nabla f}} (\mathbf x^{(t)}) \rVert^2\bigr]\\
			&\quad + \frac{5}{c^2} \sum_{t=0}^{T-1} \mathbb E \bigl[ \lVert {\mathbf{y}}^{(t+1)} - \bar {\mathbf{y}}^{(t+1)} \rVert_F^2\bigr] - \sum_{t=0}^{T-1} \mathbb E \bigl[ \lVert {\bar m}^{(t)} \rVert^2\bigr].
		\end{aligned}
	\end{equation}
\end{lemma}
\begin{proof}
See Appendix~\ref{appendixB.1}. \qed
\end{proof}

Lemma~\ref{lemma2} shows that we need to bound the consensus and average momentum estimation errors, as well as the tracking error highlighted above. We first bound the consensus error using the mixing properties of $A$.

\begin{lemma} \label{lemma3}
	Suppose Assumptions~\ref{assumption1}--\ref{assumption4} hold. Then the following inequality holds
	\begin{equation} \label{lemmaeq3}
		\begin{aligned}
			&\sum_{t=0}^{T-1} \mathbb E \bigl[ \lVert {\mathbf{x}}^{(t+1)} - \bar {\mathbf{x}}^{(t+1)} \rVert_F^2\bigr] \le \: 2 s_A^2 \alpha^2 \sum_{t=0}^{T-1} \mathbb E \bigl[ \lVert {\mathbf{y}}^{(t+1)} - \bar {\mathbf{y}}^{(t+1)} \rVert_F^2\bigr]\\
			&\quad + 2 n^2 s_A^2 \alpha^2 \sum_{t=0}^{T-1} \mathbb E \bigl[ \lVert\bar m^{(t)}\rVert^2\bigr].
		\end{aligned}
	\end{equation}
\end{lemma}
\begin{proof}
See Appendix~\ref{appendixB.2}. \qed
\end{proof}

Combining Lemmas~\ref{lemma2} and~\ref{lemma3} shows that the tracking error enters the descent bound both directly and through the consensus error. The contribution through the consensus error is scaled by $\alpha^2$, whereas the direct contribution is not. We therefore exploit recursive momentum to control the variance term in the tracking error bound. 

In particular, unrolling the tracking recursion and using the mixing properties of $B$ gives 
\[ \sum_{t=0}^{T-1}\mathbb E\!\bigl[\lVert\mathbf y^{(t+1)}-\bar{\mathbf y}^{(t+1)}\rVert_F^2\bigr]\le s_B^2\sum_{t=0}^{T-1}\mathbb E\!\bigl[\lVert\mathbf m^{(t)}-\mathbf m^{(t-1)}\rVert_F^2\bigr]. \]
For SPP, the corresponding tracking error bound involves differences between stochastic gradients computed in consecutive rounds. A direct estimate introduces a variance term of order $nTs_B^2\sigma^2$. In contrast, the recursive momentum in ASPP satisfies, for $t\ge1$, 
\[ \begin{aligned} m_i^{(t)}-m_i^{(t-1)}={}&\bigl(g_i(x_i^{(t)},\xi_i^{(t)})-g_i(x_i^{(t-1)},\xi_i^{(t)})\bigr)\\ &-(1-\beta)\bigl(m_i^{(t-1)}-g_i(x_i^{(t-1)},\xi_i^{(t)})\bigr). \end{aligned} \]
Here, the first term is controlled by the distance between consecutive iterates under mean-squared smoothness, while the second term is scaled by $1-\beta$. Using the recursive momentum update to bound the second term, we obtain a variance term of order $nTs_B^2(1-\beta)^2\sigma^2$, thereby reducing the variance contribution to the tracking error bound. The following lemma states the resulting tracking error bound.

\begin{lemma} \label{lemma4}
	Suppose Assumptions~\ref{assumption1}--\ref{assumption4} hold. If $\beta\in[0,1)$, then the following inequality holds
	\begin{equation} \label{lemmaeq4}
			\begin{aligned}
			&\sum_{t=0}^{T-1}\mathbb E\bigl[\lVert\mathbf y^{(t+1)}-\bar{\mathbf y}^{(t+1)}\rVert_F^2\bigr] \le{}9s_B^2L^2\sum_{t=0}^{T-1}\mathbb E\bigl[\lVert\mathbf x^{(t+1)}-\mathbf x^{(t)}\rVert_F^2\bigr]\\
			&\quad +9nT(1-\beta)^2s_B^2\sigma^2 +\frac{5ns_B^2\sigma^2}{b_0} +4ns_B^2\Delta_fL.
		\end{aligned}
	\end{equation}
\end{lemma}
\begin{proof}
See Appendix~\ref{appendixB.3}. \qed
\end{proof}

We next bound the average momentum estimation error. Using the recursive momentum update, we obtain a bound with a variance term of order $T(1-\beta)\sigma^2/n$. The following lemma states the result.

\begin{lemma} \label{lemma5}
	Suppose Assumptions~\ref{assumption1}--\ref{assumption4} hold. If $\beta\in[0,1)$, then the following inequality holds
	\begin{equation} \label{lemmaeq5}
		\begin{aligned}
			&\sum_{t=0}^{T-1}\mathbb E\bigl[\lVert\bar m^{(t)}-\overline{\nabla f}(\mathbf x^{(t)})\rVert^2\bigr] \le{}\frac{2 L^2}{n^2(1 - \beta)} \sum_{t=0}^{T-1} \mathbb E \bigl[\lVert {\mathbf x}^{(t+1)} - {\mathbf x}^{(t)} \rVert_F^2\bigr]\\
			&\quad + \frac{2T(1-\beta) \sigma^2}{n} + \frac{\sigma^2}{n b_0 (1 - \beta)}.
		\end{aligned}
	\end{equation}
\end{lemma}
\begin{proof}
See Appendix~\ref{appendixB.4}. \qed
\end{proof}

The bounds in Lemmas~\ref{lemma4} and~\ref{lemma5} still involve $\sum_{t=0}^{T-1}\allowbreak\mathbb E\bigl[\lVert\mathbf x^{(t+1)}-\mathbf x^{(t)}\rVert_F^2\bigr]$. Expanding the decision variable recursion yields, for $t\ge0$,
\[ \begin{aligned} \mathbf x^{(t+1)}-\mathbf x^{(t)} ={}&-c\alpha\mathds1\bar m^{(t)}-\alpha\mathds1\pi_A^\top\bigl(\mathbf y^{(t+1)}-\bar{\mathbf y}^{(t+1)}\bigr)\\ &-\alpha\sum_{l=0}^{t}\bigl(A^{t-l}-\mathds1\pi_A^\top\bigr)\bigl(\mathbf y^{(l+1)}-\mathbf y^{(l)}\bigr). \end{aligned} \]
We bound the first term directly and control the remaining terms using the tracking recursion and the mixing properties of $A$ and $B$. Treating these terms separately preserves a coefficient of order $nc^2\alpha^2$ for the average momentum term, without an additional $h^2$ factor. As explained below, this is important for establishing network-independent asymptotic linear speedup. The following lemma gives the resulting bound.

\begin{lemma} \label{lemma6}
	Suppose Assumptions~\ref{assumption1}--\ref{assumption4} hold. If $\beta\in[0,1)$ and $0<\alpha\le \frac{1}{20hL}$, then the following inequality holds
	\begin{equation} \label{lemmaeq6}
		\begin{aligned}
			&\sum_{t=0}^{T-1}\mathbb E\bigl[\lVert\mathbf x^{(t+1)}-\mathbf x^{(t)}\rVert_F^2\bigr] \le6nc^2\alpha^2\sum_{t=0}^{T-1}\mathbb E \bigl[\lVert\bar m^{(t)}\rVert^2\bigr] +144n\alpha^2h^2\Delta_fL\\
			&\quad +\frac{180n\alpha^2h^2\sigma^2}{b_0} +324nT\alpha^2h^2(1-\beta)^2\sigma^2.
		\end{aligned}
	\end{equation}
\end{lemma}
\begin{proof}
See Appendix~\ref{appendixB.5}. \qed
\end{proof}

We now substitute the bound in Lemma~\ref{lemma6} into Lemmas~\ref{lemma4} and~\ref{lemma5} and combine the resulting bounds with Lemmas~\ref{lemma2} and~\ref{lemma3}. The resulting descent inequality still contains a term involving $\sum_{t=0}^{T-1}\mathbb E\bigl[\lVert\bar m^{(t)}\rVert^2\bigr]$, whose coefficient is nonpositive under the step size conditions stated below. Dropping this nonpositive term yields the following convergence result.

\begin{samepage}
\begin{theorem}
	Suppose\label{theorem1} Assumptions~\ref{assumption1}--\ref{assumption4} hold. If
	\[
	\begin{gathered}
	0<\alpha\le\min \Bigl\{\frac{1}{4cL},\frac{1}{30hL},\frac{\sqrt{n(1-\beta)}}{20cL} \Bigr\},\\
	\beta \in [0, 1), \quad \text{and}\quad b_0\in\mathbb N_+,
	\end{gathered}
	\]
	then the following inequality holds
	\begin{equation} \label{theoremeq1}
		\begin{aligned}
			&\frac{1}{T}\sum_{t=0}^{T-1} \mathbb E\bigl[\lVert\nabla f(\hat x^{(t)})\rVert^2\bigr] \le{} \frac{4\Delta_f}{c\alpha T} +\frac{8(1-\beta)\sigma^2}{n} +\frac{80h^2(1-\beta)^2\sigma^2}{c^2}\\
			&\quad +\frac{4\sigma^2}{nb_0(1-\beta)T} +\frac{45h^2\sigma^2}{c^2b_0T} +\frac{36h^2\Delta_fL}{c^2T}.
		\end{aligned}
	\end{equation}
\end{theorem}
\begin{proof}
See Appendix~\ref{appendixB.6}. \qed
\end{proof}
\end{samepage}

Apart from initialization terms, the variance terms in Theorem~\ref{theorem1} are of orders $(1-\beta)\sigma^2/n$ and $h^2(1-\beta)^2\sigma^2/c^2$. Here, the latter contains an additional factor of $1-\beta$ and is therefore dominated by the former for sufficiently small $1-\beta$. However, decreasing $1-\beta$ also tightens the step size restriction, so the choice of $1-\beta$ must also account for the term $\Delta_f/(c\alpha T)$.

In particular, taking $\alpha=\min\Bigl\{\frac{1}{4cL},\allowbreak\frac{1}{30hL},\allowbreak\frac{\sqrt{n(1-\beta)}}{20cL}\Bigr\}$ gives 
\begin{equation} \label{theoremeq12}
	\frac{\Delta_f}{c\alpha T}\le\frac{20\Delta_fL}{T\sqrt{n(1-\beta)}}+\frac{(4+30h/c)\Delta_fL}{T}. 
\end{equation}
Here, the first term corresponds to the step-size restriction $c\alpha\le\frac{\sqrt{n(1-\beta)}}{20L}$, which is independent of the network-dependent parameter $h$. This is because, in Lemma~\ref{lemma6}, the coefficient of the average momentum term remains $6nc^2\alpha^2$, with no additional factor of $h^2$. We choose $1-\beta=n^{1/3}/T^{2/3}$ to balance the first term in \eqref{theoremeq12} with the variance term $(1-\beta)\sigma^2/n$, giving a combined contribution of order $(\Delta_fL+\sigma^2)/(n^{2/3}T^{2/3})$ with no additional network-dependent parameters. Together with an initial batch size $b_0$ chosen to control the initialization terms, these parameter choices yield the following corollary.

\begin{samepage}
\begin{corollary}
	Suppose\label{corollary1} Assumptions~\ref{assumption1}--\ref{assumption4} hold and $T\ge\sqrt n$. If
	\[
	\begin{gathered}
	\alpha=\min\Bigl\{\frac{1}{4cL},\frac{1}{30hL},\frac{n^{2/3}}{20\,cLT^{1/3}}\Bigr\},\\
	1-\beta=\frac{n^{1/3}}{T^{2/3}}, \quad \text{and}\quad b_0=\Bigl\lceil\frac{T^{1/3}}{n^{2/3}}\Bigr\rceil,
	\end{gathered}
	\]
	then the following inequality holds
	\begin{equation} \label{corollaryeq1}
		\begin{aligned}
			&\frac{1}{T}\sum_{t=0}^{T-1}\mathbb E\bigl[\lVert\nabla f(\hat x^{(t)})\rVert^2\bigr] \lesssim{}\frac{\Delta_fL+\sigma^2}{n^{2/3}T^{2/3}} +\frac{n^{2/3}h^2\sigma^2}{c^2T^{4/3}} +\frac{(c^2+h^2)\Delta_fL}{c^2T}.
		\end{aligned}
	\end{equation}
\end{corollary}
\end{samepage}

Corollary~\ref{corollary1} implies that ASPP achieves $\epsilon$-stationarity within \begin{equation} \label{corollaryeq12}
	\mathcal O\Bigl(\frac{(\Delta_fL+\sigma^2)^{3/2}}{n\epsilon^3}+\frac{\sqrt n\,(h^2\sigma^2/c^2)^{3/4}}{\epsilon^{3/2}}+\frac{(c^2+h^2)\Delta_fL}{c^2\epsilon^2}\Bigr)
\end{equation} 
communication rounds. Including initialization, each agent performs $b_0+2(T-1)=\mathcal O(T)$ stochastic gradient evaluations over $T$ communication rounds. Thus, the per-agent stochastic gradient complexity has the same order as the communication complexity.

\begin{remark} \label{remark1}
	The first term in the complexity bound \eqref{corollaryeq12} scales as $n^{-1}$ and is independent of the network-dependent parameters $c$ and $h$. The remaining two terms depend on $c$ and $h$ but are of lower order in $\epsilon$, so the first term dominates for sufficiently small $\epsilon$. A sufficient condition is 
	\begin{equation} \label{remarkeq1}
		\epsilon\lesssim\frac{\sqrt{\Delta_fL+\sigma^2}}{n(1+h^2/c^2)}. 
	\end{equation} 
	In this regime, ASPP achieves network-independent communication and per-agent stochastic gradient complexities of $\mathcal O(n^{-1}\epsilon^{-3})$, establishing asymptotic linear speedup with respect to the number of agents.
\end{remark}

\begin{remark} \label{remark2}
	ASPP achieves faster convergence than SPP under mean-squared smoothness, reducing the per-agent stochastic gradient complexity from $\mathcal O(n^{-1}\epsilon^{-4})$~\citep{liang2025linear} to $\mathcal O(n^{-1}\epsilon^{-3})$ for sufficiently small $\epsilon$ satisfying  \eqref{remarkeq1}. The corresponding total stochastic gradient complexity is $\mathcal O(\epsilon^{-3})$. This matches the $\epsilon^{-3}$ dependence of the lower bound on stochastic gradient complexity for nonconvex optimization under mean-squared smoothness~\citep{arjevani2023lower}.
\end{remark}

\section{Accelerated Stochastic Push--Pull with Local Updates}

In this section, we develop a local-update extension of ASPP. We first present its algorithmic design and then establish its convergence guarantees.

\subsection{Algorithm Design}
To further reduce communication overhead, a common strategy in federated learning is to perform multiple local updates between communication rounds \citep{fedavg,SCAFFOLD,cheng2024momentum,zhang2026non}. Motivated by this strategy, we develop a local-update variant of ASPP, termed Local Accelerated Stochastic Push--Pull (LASPP), as presented in Algorithm~\ref{alg2}.  

LASPP uses the same initialization as ASPP. In each communication round $t\ge1$, each agent~$i$ sets $x_i^{(t,0)}=x_i^{(t)}$ and performs $K$ local updates, indexed by $k=0,\ldots,K-1$, without communicating with its neighbors. At each local step~$k$, agent~$i$ constructs the recursive momentum estimator $m_i^{(t,k)}$ and uses communication-free counterparts of the ASPP recursions to form the local update direction $y_i^{(t,k+1)}$ and update $x_i^{(t,k+1)}$ with step size $\eta>0$.

After the $K$ local updates, agent~$i$ constructs $m_i^{(t)}$ using the accumulated local displacement. This construction yields
\begin{equation}\label{eq4}
	m_i^{(t)}=\frac{1}{K}\sum_{k=0}^{K-1}m_i^{(t,k)}.
\end{equation}
Agent $i$ then communicates with its neighbors and updates $y_i^{(t+1)}$ and $x_i^{(t+1)}$ using the same Push--Pull recursions as in ASPP. When $K=1$, LASPP reduces to ASPP.

\begin{algorithm}
	\caption{Local Accelerated Stochastic Push--Pull}
	\begin{algorithmic}[1]\label{alg2}
		\normalsize
		\setlength{\itemsep}{0.2em}
		\STATE \textbf{Initialization:} $x_i^{(0)} = x^{(0)}$, $y_i^{(0)} = \mathbb{0}$, $m_i^{(-1)} = \mathbb{0}$
		\STATE \hspace{\algorithmicindent}$m_i^{(0)} = \frac{1}{b_0} \sum_{r=1}^{b_0} g_i(x_i^{(0)}, \xi_{i,r}^{(0)})$
		\STATE \hspace{\algorithmicindent}$y_i^{(1)} = \sum_{j=1}^{n} b_{ij} y_j^{(0)} + m_i^{(0)} - m_i^{(-1)}$
		\STATE \hspace{\algorithmicindent}$x_i^{(1)} = \sum_{j=1}^{n} a_{ij} x_j^{(0)} - \alpha y_i^{(1)}$
		\FOR{ $t = 1, 2, \ldots, T-1$}
		\setlength{\itemsep}{0.2em}
		\STATE $x_i^{(t,0)} = x_i^{(t)}$
		\FOR{$k = 0, \dots, K - 1$ }
		\setlength{\itemsep}{0.2em}
		\STATE $
		\!\! m_i^{(t,k)} \!\! =\! g_i(x_i^{(t,k)}\!\!,  \xi_i^{(t,k)}) + \beta \bigl(m_i^{(t-1)} \!- g_i(x_i^{(t-1)}\!\!, \xi_i^{(t,k)})\bigr)
		$
		\STATE $\!\! y_i^{(t,k+1)} = y_i^{(t)} + m_i^{(t,k)} - m_i^{(t-1)}$
		\STATE $\!\! x_i^{(t,k+1)} = x_i^{(t,k)} - \eta y_i^{(t,k+1)}$
		\ENDFOR
		\STATE
		$
		m_i^{(t)} \! = \frac{1}{\eta K}(x_i^{(t)} - x_i^{(t,K)}) - y_i^{(t)} + m_i^{(t-1)} 
		$
		\STATE $y_i^{(t+1)} = \sum_{j=1}^{n} b_{ij} y_j^{(t)} + m_i^{(t)} - m_i^{(t-1)}$
		\STATE $x_i^{(t+1)} = \sum_{j=1}^{n} a_{ij} x_j^{(t)} - \alpha y_i^{(t+1)}$
		\ENDFOR
	\end{algorithmic}
\end{algorithm}

As shown in Algorithm~\ref{alg2}, LASPP is not obtained by simply skipping communication rounds in ASPP. Its update rules are instead designed to accommodate multiple local steps between successive communication rounds. The main features of these update rules are discussed in the following two remarks.

\begin{remark} \label{remark3}
	An alternative interpretation of the local updates follows from rewriting Line~9 as
	\[
	y_i^{(t,k+1)} = m_i^{(t,k)} + \bigl(y_i^{(t)}-m_i^{(t-1)}\bigr).
	\]
	This form reveals a correction mechanism analogous to that used in SCAFFOLD~\citep{SCAFFOLD}. SCAFFOLD uses the difference between estimates of the global and local update directions to correct local updates and mitigate client drift. LASPP applies an analogous correction using $y_i^{(t)}-m_i^{(t-1)}$, where $y_i^{(t)}$ incorporates network information and $m_i^{(t-1)}$ is the previous local momentum estimate. With this correction, each agent forms its local direction by adding the momentum difference $m_i^{(t,k)}-m_i^{(t-1)}$ to $y_i^{(t)}$, rather than using the current local momentum alone. In the analysis, the tracking recursion and the recursive momentum structure can then be used together to bound the deviations of the local iterates from $x_i^{(t)}$.
\end{remark}

\begin{remark} \label{remark4}
	Combining \eqref{eq4} with Line~8, 
	$m_i^{(t)}$ can be written as
	\begin{equation}\label{eq5} 
		\begin{aligned}
			&m_i^{(t)} = \frac{1}{K}\sum_{k=0}^{K-1}g_i(x_i^{(t,k)}, \xi_i^{(t,k)}) + \beta \Bigl(m_i^{(t-1)} - \frac{1}{K}\sum_{k=0}^{K-1}g_i(x_i^{(t-1)}, \xi_i^{(t,k)})\Bigr).
		\end{aligned}
	\end{equation}
	This expression shows that $m_i^{(t)}$ retains the recursive momentum structure. More importantly, compared with \eqref{eq3}, $m_i^{(t)}$ in \eqref{eq5} is constructed using averages of stochastic gradients over $K$ local steps. Intuitively, this averaging helps mitigate the effect of stochastic gradient noise. However, the first average in \eqref{eq5} uses gradients evaluated at successive local iterates $x_i^{(t,k)}$ rather than all at the starting point $x_i^{(t)}$. The subsequent analysis shows how to control the deviations of these local iterates from $x_i^{(t)}$ while preserving the noise reduction from gradient averaging, leading to improved communication complexity.
\end{remark}

\subsection{Convergence Analysis}

In this subsection, we analyze LASPP and characterize the effect of local updates on convergence. We retain the notation of Section~3.2 and define 
\[
\mathbf x^{(t,k)}
:=[x_1^{(t,k)},x_2^{(t,k)},\ldots,x_n^{(t,k)}]^\top
\in\mathbb R^{n\times p}.
\] 
Since LASPP preserves the communication updates and initialization of ASPP, the descent inequality in Lemma~\ref{lemma2} and the consensus error bound in Lemma~\ref{lemma3} remain valid under their stated conditions. We therefore need bounds on the tracking and average momentum estimation errors for LASPP. To establish the communication benefit of local updates, these bounds must capture the reduction in sampling noise from gradient averaging in \eqref{eq5} while accounting for the errors introduced as the local iterates $x_i^{(t,k)}$ deviate from the starting point $x_i^{(t)}$ of each round.

As in the analysis of ASPP, we first bound the tracking error by controlling the differences $\mathbf m^{(t)}-\mathbf m^{(t-1)}$. From \eqref{eq5}, we have, for $t\ge1$,
\begin{equation} \label{lemmaeq71}
	\begin{aligned}
		m_i^{(t)}-m_i^{(t-1)}
		={}&\frac1K\sum_{k=0}^{K-1}\bigl(g_i(x_i^{(t,k)},\xi_i^{(t,k)})-g_i(x_i^{(t-1)},\xi_i^{(t,k)})\bigr)\\
		&-(1-\beta)\Bigl(m_i^{(t-1)}-\frac1K\sum_{k=0}^{K-1}g_i(x_i^{(t-1)},\xi_i^{(t,k)})\Bigr).
	\end{aligned}
\end{equation}
Bounding the first term under mean-squared smoothness introduces both $\lVert x_i^{(t)}-x_i^{(t-1)}\rVert^2$ and the local drift $\lVert x_i^{(t,k)}-x_i^{(t)}\rVert^2$. In the second term, the stochastic gradients are evaluated at the same point using $K$ independent samples, so their average has sampling variance at most $\sigma^2/K$. Using this variance bound and the recursive momentum structure, we obtain a variance term of order $nTs_B^2(1-\beta)^2\sigma^2/K$ in the tracking error bound. The following lemma states the resulting bound.

\begin{lemma} \label{lemma7}
	Suppose Assumptions~\ref{assumption1}--\ref{assumption4} hold and $\beta\in[0,1)$. Then the following inequality holds
	\begin{equation} \label{lemmaeq7}
		\begin{aligned}
			&\sum_{t=0}^{T-1}\mathbb E\bigl[\lVert\mathbf y^{(t+1)}-\bar{\mathbf y}^{(t+1)}\rVert_F^2\bigr] \le{}15 s_B^2 L^2\sum_{t=0}^{T-1}\mathbb E\bigl[\lVert\mathbf x^{(t+1)}-\mathbf x^{(t)}\rVert_F^2\bigr]\\
			&\quad +\frac{18 s_B^2 L^2}{K}\sum_{t=1}^{T-1}\sum_{k=0}^{K-1}\mathbb E\bigl[\lVert\mathbf x^{(t,k)}-\mathbf x^{(t)}\rVert_F^2\bigr]\\
			&\quad +\frac{12nT s_B^2 (1-\beta)^2\sigma^2}{K} + \frac{5ns_B^2\sigma^2}{b_0} +4ns_B^2\Delta_fL.
		\end{aligned}
	\end{equation}
\end{lemma}
\begin{proof}
See Appendix~\ref{appendixC.1}. \qed
\end{proof}

We next bound the average momentum estimation error. By analyzing the recursive momentum update, we obtain the following bound with a variance term of order $T(1-\beta)\sigma^2/(nK)$.

\begin{lemma} \label{lemma8}
	Suppose Assumptions~\ref{assumption1}--\ref{assumption4} hold and $\beta\in[0,1)$. Then the following inequality holds
	\begin{equation} \label{lemmaeq8}
		\begin{aligned}
			&\sum_{t=0}^{T-1}\mathbb E\bigl[\lVert\bar m^{(t)} - \overline {\nabla f}(\mathbf x^{(t)})\rVert^2\bigr] \le \frac{4 L^2}{n K (1-\beta)^2} \sum_{t=1}^{T-1} \sum_{k=0}^{K-1} \mathbb E \bigl[ \lVert {\mathbf x}^{(t,k)} - {\mathbf x}^{(t)} \rVert_F^2 \bigr]\\
			&+ \frac{3 L^2}{n^2 K (1-\beta)} \sum_{t=0}^{T-1} \mathbb E \bigl[ \lVert {\mathbf x}^{(t+1)} - {\mathbf x}^{(t)} \rVert_F^2\bigr] + \frac{3T(1-\beta) \sigma^2}{n K} + \frac{\sigma^2}{n b_0 (1-\beta)},
		\end{aligned}
	\end{equation}
\end{lemma}
\begin{proof}
See Appendix~\ref{appendixC.2}. \qed
\end{proof}

The bounds in Lemmas~\ref{lemma7} and~\ref{lemma8} still involve the term $\sum_{t=0}^{T-1}\mathbb E\bigl[\lVert\mathbf x^{(t+1)}-\mathbf x^{(t)}\rVert_F^2\bigr]$, as well as the local-drift term $\sum_{t=1}^{T-1}\sum_{k=0}^{K-1}\mathbb E\bigl[\lVert\mathbf x^{(t,k)}-\mathbf x^{(t)}\rVert_F^2\bigr]$. We first bound the former using the decomposition employed in the proof of Lemma~\ref{lemma6}. The following lemma gives the resulting bound.

\begin{lemma} \label{lemma9}
	Suppose Assumptions~\ref{assumption1}--\ref{assumption4} hold. If $\beta\in[0,1)$ and $0<\alpha\le\frac{1}{25hL}$, then the following inequality holds
	\begin{equation} \label{lemmaeq9}
		\begin{aligned}
			&\sum_{t=0}^{T-1}\mathbb E\bigl[\lVert\mathbf x^{(t+1)}-\mathbf x^{(t)}\rVert_F^2\bigr] \le{}6nc^2\alpha^2\sum_{t=0}^{T-1}\mathbb E\bigl[\lVert\bar m^{(t)}\rVert^2\bigr]\\
			&\quad +\frac{648\alpha^2h^2L^2}{K}\sum_{t=1}^{T-1}\sum_{k=0}^{K-1}\mathbb E\bigl[\lVert\mathbf x^{(t,k)}-\mathbf x^{(t)}\rVert_F^2\bigr]\\
			&\quad +\frac{432nT\alpha^2h^2(1-\beta)^2\sigma^2}{K} +\frac{180n\alpha^2h^2\sigma^2}{b_0} +144n\alpha^2h^2\Delta_fL.
		\end{aligned}
	\end{equation}
\end{lemma}
\begin{proof}
See Appendix~\ref{appendixC.3}. \qed
\end{proof}

We now exploit the structure of the local updates to bound the local drift. For $t\ge1$, unrolling the local recursion for $x_i^{(t,k)}$ gives
\[
x_i^{(t,k)}-x_i^{(t)}
=-\eta k y_i^{(t)}
-\eta\sum_{r=0}^{k-1}\bigl(m_i^{(t,r)}-m_i^{(t-1)}\bigr).
\]
We bound the first term through the tracking error and average momentum, and the second using an expansion analogous to \eqref{lemmaeq71}. Combining these estimates and averaging over the $K$ local steps, we obtain the following local-drift bound with a variance term of order $nT\eta^2Ks_B^2(1-\beta)^2\sigma^2$.

\begin{lemma} \label{lemma10}
	Suppose Assumptions~\ref{assumption1}--\ref{assumption4} hold and $\beta\in[0,1)$. If $0<\alpha\le\frac{1}{25hL}$ and $0<\eta\le\frac{1}{20s_BKL}$, then the following inequality holds
	\begin{equation} \label{lemmaeq10}
		\begin{aligned}
			&\frac{1}{K}\sum_{t=1}^{T-1}\sum_{k=0}^{K-1}\mathbb E\bigl[\lVert\mathbf x^{(t,k)}-\mathbf x^{(t)}\rVert_F^2\bigr] \le{}10n^2\eta^2K^2\sum_{t=0}^{T-1}\mathbb E\bigl[\lVert\bar m^{(t)}\rVert^2\bigr]\\
			&\quad +66n\eta^2K^2s_B^2\Delta_fL +216nT\eta^2Ks_B^2(1-\beta)^2\sigma^2 +\frac{86n\eta^2K^2s_B^2\sigma^2}{b_0}.
		\end{aligned}
	\end{equation}
\end{lemma}
\begin{proof}
See Appendix~\ref{appendixC.4}. \qed
\end{proof}

Apart from initialization terms, the bound in Lemma~\ref{lemma10} consists of a term proportional to $(\eta K)^2\sum_{t=0}^{T-1}\mathbb E\bigl[\lVert\bar m^{(t)}\rVert^2\bigr]$ and a variance term proportional to $(\eta K)^2(1-\beta)^2/K$. A suitable restriction on $\eta K$ therefore allows the average momentum term introduced by local drift to be absorbed by the negative term in the descent inequality, while preserving the $1/K$ factor in the variance contribution to the convergence bound. Moreover, the factor $(1-\beta)^2$ shows that recursive momentum also controls the variance contribution to local drift.

We now combine the bounds in Lemmas~\ref{lemma3} and~\ref{lemma7}--\ref{lemma10} with the descent inequality in Lemma~\ref{lemma2}. The resulting inequality contains the term $\sum_{t=0}^{T-1}\mathbb E\bigl[\lVert\bar m^{(t)}\rVert^2\bigr]$, whose coefficient is nonpositive under the step size conditions stated below. Dropping this nonpositive term yields the following convergence result.

\begin{samepage}
\begin{theorem}
	Suppose\label{theorem2} Assumptions~\ref{assumption1}--\ref{assumption4} hold. If
	\[
	\begin{gathered}
	0<\alpha\le\min\Bigl\{\frac{1}{4cL},\frac{1}{50hL},\frac{\sqrt{nK(1-\beta)}}{20cL}\Bigr\},\\
	0<\eta\le\min\Bigl\{\frac{c}{100ns_BKL},\frac{1-\beta}{100\sqrt n\,KL}\Bigr\},\\
	\beta \in [0, 1), \quad \text{and}\quad b_0\in\mathbb N_+,
	\end{gathered}
	\]
	then the following inequality holds
	\begin{equation} \label{theoremeq2}
		\begin{aligned}
			&\frac{1}{T}\sum_{t=0}^{T-1} \mathbb E\bigl[\lVert\nabla f(\hat x^{(t)})\rVert^2\bigr] \le{} \frac{4\Delta_f}{c\alpha T} +\frac{12(1-\beta)\sigma^2}{nK} +\frac{105h^2(1-\beta)^2\sigma^2}{c^2K}\\
			&\quad +\frac{4\sigma^2}{nb_0(1-\beta)T} +\frac{45h^2\sigma^2}{c^2b_0T} +\frac{36h^2\Delta_fL}{c^2T}.
		\end{aligned}
	\end{equation}
\end{theorem}
\begin{proof}
See Appendix~\ref{appendixC.5}. \qed
\end{proof}
\end{samepage}

The bound in Theorem~\ref{theorem2} retains the factor of $1/K$ in the variance terms after accounting for local drift. Apart from initialization terms, the variance terms are of orders $(1-\beta)\sigma^2/(nK)$ and $h^2(1-\beta)^2\sigma^2/(c^2K)$. Meanwhile, the step size upper bound involving $1-\beta$ is larger by a factor of $\sqrt K$ than its ASPP counterpart, giving $c\alpha\leq\sqrt{nK(1-\beta)}/(20L)$ without an additional factor involving $h$.

Following the parameter-selection argument for ASPP, we balance\newline $\frac{\Delta_fL}{T\sqrt{nK(1-\beta)}}$ with $(1-\beta)\sigma^2/(nK)$. Choosing $1-\beta=(nK)^{1/3}/T^{2/3}$ gives a network-independent contribution of order $(\Delta_fL+\sigma^2)/(nKT)^{2/3}$, gaining an additional factor of $K^{-2/3}$ over the corresponding contribution for ASPP. With suitable choices of the step sizes and initial batch size, we obtain the following corollary.

\begin{corollary}
	Suppose\label{corollary2} Assumptions~\ref{assumption1}--\ref{assumption4} hold and $T\ge\sqrt{nK}$. If
	\[
	\begin{gathered}
	\alpha=\min\Bigl\{\frac{1}{4cL},\frac{1}{50hL},\frac{n^{2/3}K^{2/3}}{20cLT^{1/3}}\Bigr\},\\
	0<\eta\le\min\Bigl\{\frac{c}{100ns_BKL},\frac{1}{100n^{1/6}LK^{2/3}T^{2/3}}\Bigr\},\\
	1 - \beta = \frac{n^{1/3} K^{1/3}}{T^{2/3}}, \quad \text{and}\quad b_0 = \Bigl\lceil\frac{K^{1/3} T^{1/3}}{n^{2/3}}\Bigr\rceil,
	\end{gathered}
	\]
	then the following inequality holds
	\begin{equation} \label{corollaryeq2}
		\begin{aligned}
			\frac{1}{T}\sum_{t=0}^{T-1}
			\!\mathbb E\bigl[\lVert\nabla f(\hat x^{(t)})\rVert^2\bigr]
			\lesssim
			\frac{\Delta_fL+\sigma^2}{n^{2/3}K^{2/3}T^{2/3}}
			+\frac{n^{2/3}h^2\sigma^2}{c^2K^{1/3}T^{4/3}}+\frac{(c^2+h^2)\Delta_fL}{c^2T}.
		\end{aligned}
	\end{equation}
\end{corollary}

Corollary~\ref{corollary2} implies that LASPP achieves $\epsilon$-stationarity within 
\begin{equation}\label{corollaryeq22} 
	\mathcal O \Bigl(
	\frac{(\Delta_fL+\sigma^2)^{3/2}}{nK\epsilon^3}
	+\frac{\sqrt n}{K^{1/4}\epsilon^{3/2}}
	\Bigl(\frac{h^2\sigma^2}{c^2}\Bigr)^{3/4}
	+\frac{(c^2+h^2)\Delta_fL}{c^2\epsilon^2}
	\Bigr)
\end{equation}
communication rounds. Compared with the corresponding bound for ASPP, the first two terms are reduced by factors of $K$ and $K^{1/4}$, respectively, capturing the communication benefit of local updates.

\begin{remark} \label{remark21}
	The first term in \eqref{corollaryeq22} scales as $n^{-1}K^{-1}$ and is independent of the network parameters $c$ and $h$. A sufficient condition for this term to dominate the remaining terms is 
	\begin{equation} \label{remarkeq21}
		\epsilon\lesssim\frac{\sqrt{\Delta_fL+\sigma^2}}{nK(1+h^2/c^2)}. 
	\end{equation}
	In this regime, LASPP has a communication complexity of $\mathcal O(n^{-1}K^{-1}\epsilon^{-3})$. Compared with ASPP, LASPP thus achieves an additional asymptotic linear speedup with respect to the number of local steps $K$, while retaining linear speedup with respect to the number of agents $n$.
\end{remark}

\begin{remark} \label{remark22}
	Including initialization, each agent performs $b_0+2K(T-1)=\mathcal O(KT)$ stochastic gradient evaluations. By \eqref{corollaryeq22}, the resulting per-agent stochastic gradient complexity has the same order as its ASPP counterpart for every integer $1\le K\le\bar K$, where
	$\bar K:=\biggl\lceil\min\Bigl\{\frac{\sqrt{\Delta_fL+\sigma^2}}{n(1+h^2/c^2)\,\epsilon},\allowbreak\frac{c^2(\Delta_fL+\sigma^2)}{n^2h^2\epsilon^2}\Bigr\}\biggr\rceil.$ In particular, choosing $K=\bar K$ allows LASPP to achieve $\epsilon$-stationarity within \[ \mathcal O \Bigl(\frac{(c^2+h^2)(\Delta_fL+\sigma^2)}{c^2\epsilon^2}+\frac{nh^2\sqrt{\Delta_fL+\sigma^2}}{c^2\epsilon}\Bigr) \] communication rounds. This improves the asymptotic communication complexity from $\mathcal O(\epsilon^{-3})$ for ASPP to $\mathcal O(\epsilon^{-2})$, while maintaining the same per-agent stochastic gradient complexity. The $\epsilon^{-2}$ dependence matches the known communication lower bound for centralized stochastic nonconvex optimization with local updates~\citep{zhang2021fedpd}.
\end{remark}

\section{Numerical Experiments}

In the numerical experiments, we consider two learning tasks on real datasets: multiclass logistic regression on MNIST and convolutional neural network (CNN) training on CIFAR-10. 

We compare the proposed algorithms with SPP. In addition, we include SCAFFOLD-MVR \citep{cheng2024momentum} as a centralized baseline, providing a reference for the performance achievable by local-update methods when a central server is available. SCAFFOLD-MVR is an accelerated variant of SCAFFOLD based on recursive momentum. For the decentralized methods, performance is evaluated using the average of the local models. 

For all decentralized algorithms, we use a strongly connected Erdős--Rényi digraph $\mathcal G$ generated with edge probability $0.3$ and augmented with a self-loop at each node. We take $\mathcal G_A=\mathcal G_B=\mathcal G$ and construct the mixing weights as
\[
a_{ij}=
\begin{cases}
	\dfrac{1}{|\mathcal N_{A,i}^{\mathrm{in}}|}, & \text{if } (j,i)\in\mathcal E,\\
	0, & \text{otherwise},
\end{cases}
\quad \quad \quad
b_{ij}=
\begin{cases}
	\dfrac{1}{|\mathcal N_{B,j}^{\mathrm{out}}|}, & \text{if } (j,i)\in\mathcal E,\\
	0, & \text{otherwise}.
\end{cases}
\]

\subsection{Nonconvex Logistic Regression}

We first consider a multiclass logistic regression problem on the MNIST dataset. MNIST contains $60{,}000$ training images from $C=10$ handwritten digit classes. To simulate heterogeneous data distributions across agents, we partition the training data among the agents by Dirichlet sampling with parameter $0.9$. The local dataset of agent~$i$ is denoted by $\{(z_{i,s},\ell_{i,s})\}_{s=1}^{N_i}$, where $N_i$ is the number of local samples, $z_{i,s}\in\mathbb R^d$ is the $d$-dimensional image feature vector, and $\ell_{i,s}\in\{0,1\}^{C}$ is the one-hot encoded label. Let $w_r\in\mathbb R^d$ denote the model parameters for class~$r$, and write $x=[w_1^\top,\ldots,w_C^\top]^\top\in\mathbb R^{Cd}$. The local loss function $f_i$ in problem~\eqref{eq1} is then given by
\[
f_i(x)=\frac{1}{N_i}\sum_{s=1}^{N_i}\Biggl[-\sum_{r=1}^{C}\ell_{i,s}^{r}\log\frac{\exp(w_r^\top z_{i,s})}{\sum_{k=1}^{C}\exp(w_k^\top z_{i,s})}\Biggr] +\rho\sum_{q=1}^{Cd}\frac{x_q^2}{1+x_q^2},
\]
where $x_q$ denotes the $q$th entry of $x$. The last term is the nonconvex Geman--McClure regularizer with regularization parameter $\rho=0.01$.

In this experiment, we set $n=16$. The mini-batch size is $16$, and $b_0=16$ is used when an initial batch is required. We set $\alpha=0.01$ for all methods and $\beta=0.9$ for momentum-based methods. For local-update methods, we set $K=10$ and $\eta=0.001$.

\begin{figure}[htbp]
	\centering
	\includegraphics[width=0.68\linewidth]{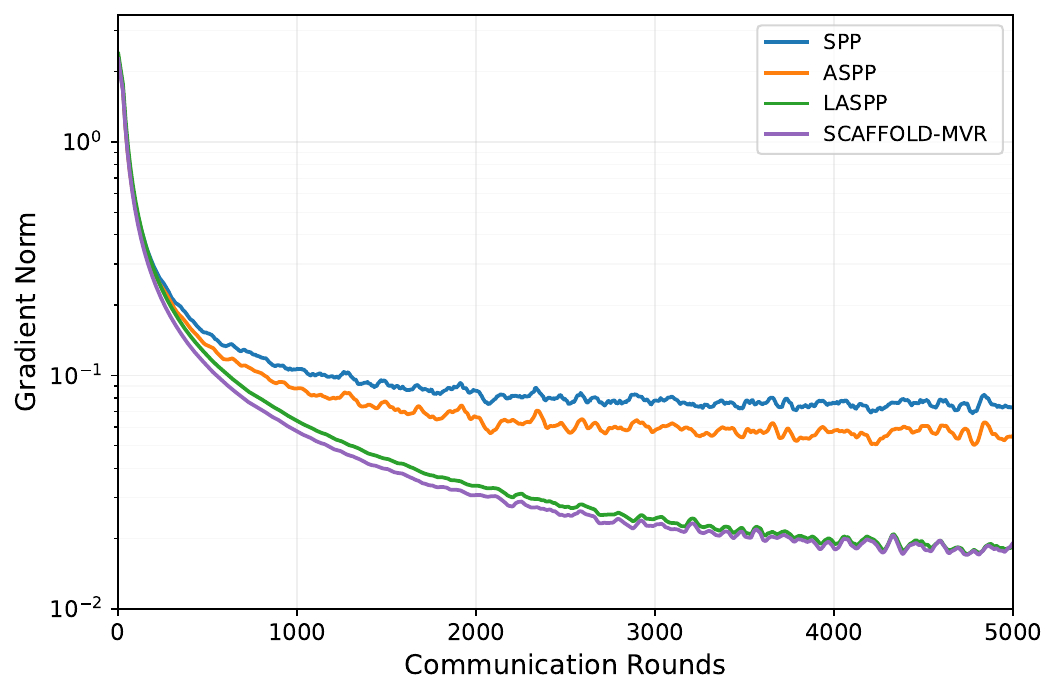}
	\caption{Comparison of different methods for nonconvex logistic regression on MNIST.}
	\label{fig3.1}
\end{figure}

Fig.~\ref{fig3.1} compares the four methods for the nonconvex logistic regression problem, using the global gradient norm as the performance metric. As shown in the figure, ASPP achieves faster convergence than SPP. LASPP further improves upon ASPP by performing multiple local updates between communication rounds. Moreover, LASPP exhibits performance comparable to that of the centralized baseline SCAFFOLD-MVR.

\begin{figure}[htbp]
	\centering
	\includegraphics[width=0.68\linewidth]{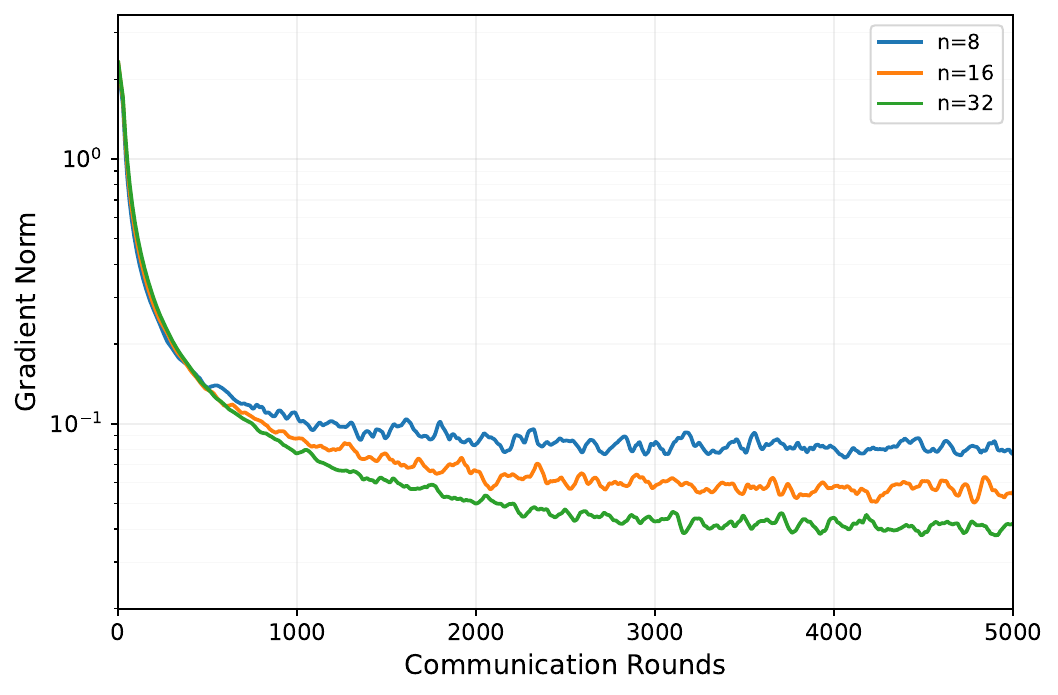}
	\caption{Performance of ASPP with different numbers of agents, $n=8,16,32$, for nonconvex logistic regression on MNIST.}
	\label{fig3.2}
\end{figure}

We further examine the convergence benefits of increasing the number of agents or local update steps. Fig.~\ref{fig3.2} reports the performance of ASPP with $n=8,16,32$ agents, while Fig.~\ref{fig3.3} reports the performance of LASPP with $K=5,10,20$ local update steps. In both experiments, the benefits of increasing $n$ or $K$ become more pronounced as optimization proceeds, with lower gradient norms attained for the same number of communication rounds. These observations are consistent with the asymptotic linear speedup established in our analysis.

\begin{figure}[htbp]
	\centering
	\includegraphics[width=0.68\linewidth]{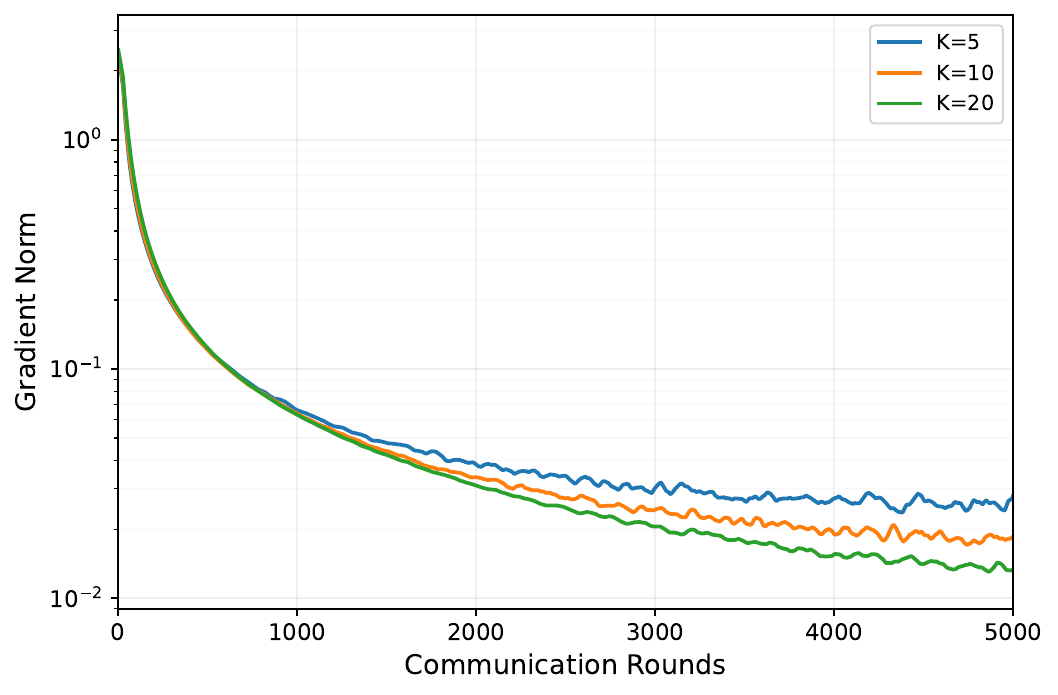}
	\caption{Performance of LASPP with different numbers of local steps, $K=5,10,20$, for nonconvex logistic regression on MNIST.}
	\label{fig3.3}
\end{figure}

\subsection{CNN Training}

We next evaluate the practical performance of the proposed algorithms by training a CNN for image classification on CIFAR-10. CIFAR-10 contains $50{,}000$ color training images from $10$ classes. The training data are distributed across agents by Dirichlet sampling with parameter $0.9$, leading to heterogeneous local datasets. The same CNN architecture is used at all agents. The CNN consists of two $3\times 3$ convolutional layers with $32$ feature maps, each followed by ReLU activation and $2\times 2$ max pooling, and three fully connected layers with $256$, $64$, and $10$ units. The training loss is the cross-entropy loss.

For this experiment, we set $n=8$ and use a mini-batch size of $16$. When an initial batch is required, we take $b_0=16$. We set $\alpha=0.003$ for all methods and $\beta=0.8$ for momentum-based methods. For local-update methods, we set $K=5$ and $\eta=0.001$.

Fig.~\ref{fig3.4} illustrates the test accuracy curves of the four methods. As shown in the figure, ASPP consistently outperforms SPP. LASPP achieves higher accuracy than ASPP with fewer communication rounds, and its test accuracy is close to that of the centralized baseline SCAFFOLD-MVR.

\begin{figure}[htbp]
	\centering
	\includegraphics[width=0.68\linewidth]{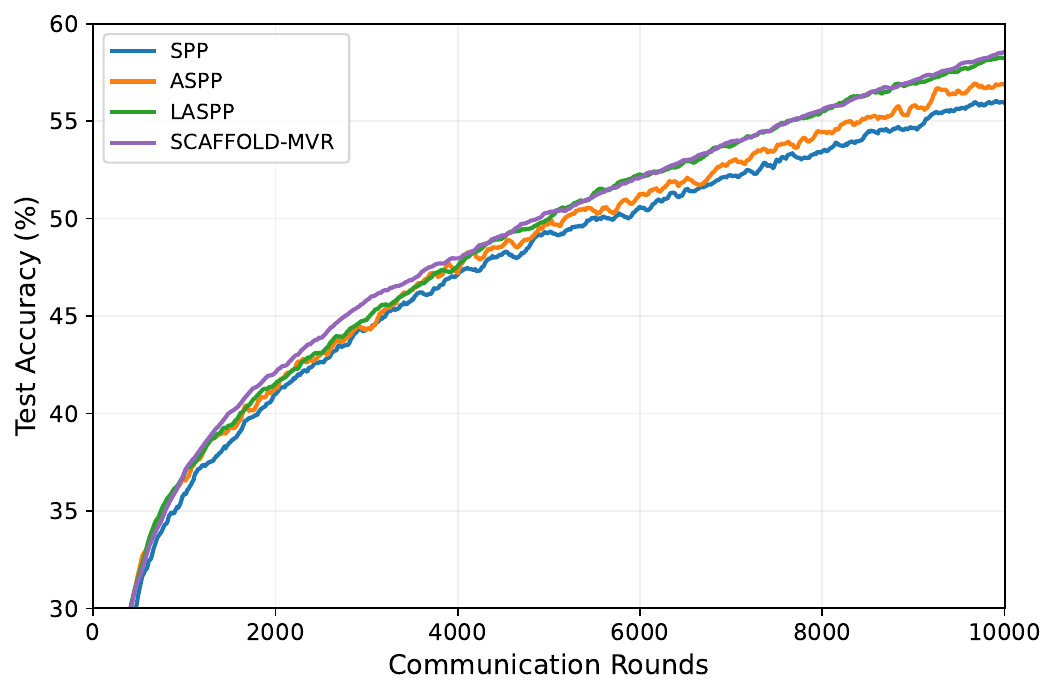}
	\caption{Comparison of different methods for CNN training on CIFAR-10.}
	\label{fig3.4}
\end{figure}

\section{Conclusion}

In this paper, we developed ASPP and its local-update variant LASPP for decentralized stochastic nonconvex optimization over directed networks. By incorporating recursive momentum for variance reduction, ASPP improves upon the convergence rate of SPP under mean-squared smoothness and achieves asymptotic linear speedup with respect to the number of agents. LASPP further achieves asymptotic linear speedup with respect to the number of local update steps. Compared with ASPP, LASPP reduces communication complexity without increasing stochastic gradient complexity when the number of local steps is suitably chosen. The numerical results demonstrate the effectiveness of both algorithms and support our theoretical findings.

\FloatBarrier
\appendix
\numberwithin{equation}{section}
\section{Supporting Lemmas}

This section introduces several supporting lemmas.

\begin{lemma} \label{lemma11}
	Suppose $A$ is a row-stochastic matrix satisfying Assumption~\ref{assumption4}. For any sequence of matrices $\Delta^{(l)} \in \mathbb{R}^{n \times p}$, $0 \le l \le T$, we have
	\begin{equation} \label{lemmaeq11}
		\sum_{t=0}^{T} \Bigl\lVert\sum_{l=0}^{t} (A^{t-l} - \mathds 1 \pi_A^\top)\Delta^{(l)} \Bigr\rVert_F^2 \leq s_A^2 \sum_{l=0}^{T} \lVert\Delta^{(l)}\rVert_F^2.
	\end{equation}
	An analogous bound holds for a column-stochastic matrix $B$ satisfying Assumption~\ref{assumption4}, with $A^{t-l}-\mathds 1 \pi_A^\top$ replaced by $B^{t-l}-\pi_B \mathds 1^\top$ and $s_A$ by $s_B$.
\end{lemma}
\begin{proof}
The result follows from the proof of Lemma~A.1 in~\citep{liang2025linear}. \qed
\end{proof}

The following lemma bounds the initialization errors for both algorithms.

\begin{lemma} \label{lemma12}
	Suppose Assumptions~\ref{assumption1}--\ref{assumption4} hold. Then, we have
	\begin{equation} \label{lemmaeq12}
		\begin{aligned}
			\mathbb E\bigl[\lVert\bar m^{(0)} - \overline {\nabla f}(\mathbf x^{(0)})\rVert^2\bigr] \le \frac{\sigma^2}{n b_0}, \quad\quad
			\mathbb E \bigl[\lVert\mathbf{m}^{(0)} - {\nabla f}(\mathbf x^{(0)}) \rVert_F^2 \bigr] \le \frac{n \sigma^2}{b_0}.
		\end{aligned}
	\end{equation}
\end{lemma}
\begin{proof}
See Appendix~\ref{appendixD.1}. \qed
\end{proof}

The next two lemmas concern the iterates generated by Algorithm~\ref{alg1}.

\begin{lemma} \label{lemma13}
	Suppose Assumptions~\ref{assumption1}--\ref{assumption4} hold and $\beta \in [0,1)$. Then the following inequality holds
	\begin{equation} \label{lemmaeq13}
		\sum_{t=0}^{T-1} \mathbb E \bigl[\lVert {\mathbf m}^{(t)} - {\nabla f} (\mathbf x^{(t)}) \rVert_F^2\bigr] \le \: \frac{2 L^2}{1 - \beta} \sum_{t=0}^{T-1} \mathbb E \bigl[\lVert {\mathbf x}^{(t+1)} - {\mathbf x}^{(t)} \rVert_F^2\bigr]+ \frac{n \sigma^2}{b_0 (1 - \beta)} + 2nT(1-\beta) \sigma^2.
	\end{equation}
\end{lemma}
\begin{proof}
See Appendix~\ref{appendixD.2}. \qed
\end{proof}

\begin{samepage}
\begin{lemma}
	Suppose\label{lemma14} Assumptions~\ref{assumption1}--\ref{assumption4} hold and $\beta \in [0,1)$. Then the following inequality holds
	\begin{equation} \label{lemmaeq14}
			\sum_{t=1}^{T-1} \mathbb E \bigl[\lVert {\mathbf m}^{(t)} - {\mathbf m}^{(t-1)} \rVert_F^2\bigr] 
			\le 9L^2 \sum_{t=0}^{T-1} \mathbb E \bigl[\lVert {\mathbf x}^{(t+1)} - {\mathbf x}^{(t)} \rVert_F^2\bigr] + \frac{3n\sigma^2}{b_0} + 9nT(1-\beta)^2\sigma^2.
	\end{equation}
\end{lemma}
\begin{proof}
See Appendix~\ref{appendixD.3}. \qed
\end{proof}
\end{samepage}

The next two lemmas concern the iterates generated by Algorithm~\ref{alg2}.

\begin{lemma} \label{lemma15}
	Suppose Assumptions~\ref{assumption1}--\ref{assumption4} hold and $\beta\in[0,1)$. Then the following inequality holds
	\begin{equation} \label{lemmaeq15}
		\begin{aligned}
		&\sum_{t=0}^{T-1}\mathbb E\bigl[\lVert\mathbf m^{(t)}-\nabla f(\mathbf x^{(t)})\rVert_F^2\bigr] \le{}\frac{4L^2}{K(1-\beta)^2}\sum_{t=1}^{T-1}\sum_{k=0}^{K-1}\mathbb E\bigl[\lVert\mathbf x^{(t,k)}-\mathbf x^{(t)}\rVert_F^2\bigr]\\
		&\quad +\frac{3L^2}{K(1-\beta)}\sum_{t=0}^{T-1}\mathbb E\bigl[\lVert\mathbf x^{(t+1)}-\mathbf x^{(t)}\rVert_F^2\bigr] +\frac{n\sigma^2}{b_0(1-\beta)}+\frac{3nT(1-\beta)\sigma^2}{K}.
	\end{aligned}
	\end{equation}
\end{lemma}
\begin{proof}
See Appendix~\ref{appendixD.4}. \qed
\end{proof}

\begin{lemma} \label{lemma16}
	Suppose Assumptions~\ref{assumption1}--\ref{assumption4} hold and $\beta\in[0,1)$. Then the following inequality holds
	\begin{equation} \label{lemmaeq16}
		\begin{aligned}
		&\sum_{t=1}^{T-1}\mathbb E\bigl[\lVert\mathbf m^{(t)}-\mathbf m^{(t-1)}\rVert_F^2\bigr] \le{}15L^2\sum_{t=0}^{T-1}\mathbb E\bigl[\lVert\mathbf x^{(t+1)}-\mathbf x^{(t)}\rVert_F^2\bigr]\\
		&\quad +\frac{18L^2}{K}\sum_{t=1}^{T-1}\sum_{k=0}^{K-1}\mathbb E\bigl[\lVert\mathbf x^{(t,k)}-\mathbf x^{(t)}\rVert_F^2\bigr] +\frac{3n\sigma^2}{b_0}+\frac{12nT(1-\beta)^2\sigma^2}{K}.
	\end{aligned}
	\end{equation}
\end{lemma}
\begin{proof}
See Appendix~\ref{appendixD.5}. \qed
\end{proof}

\section{Proofs for ASPP}

Let $\mathcal F_t$ denote the $\sigma$-algebra generated by all samples drawn before communication round~$t$. Set $\mathcal F_0$ to be the trivial $\sigma$-algebra and define \[\begin{aligned}\mathcal F_1&:=\sigma\bigl(\xi_{i,r}^{(0)}:i\in\mathcal N,\ 1\le r\le b_0\bigr),\\ \mathcal F_{t+1}&:=\sigma\bigl(\mathcal F_t,\xi_i^{(t)}:i\in\mathcal N\bigr),\qquad t\ge1.\end{aligned}\]
For $t\ge1$, the update rules of Algorithm~\ref{alg1} imply that $\mathbf x^{(t)}$, $\mathbf y^{(t)}$, and $\mathbf m^{(t-1)}$ are $\mathcal F_t$-measurable. By Assumption~\ref{assumption1}, the samples $\{\xi_i^{(t)}\}_{i=1}^n$ are mutually independent and independent of $\mathcal F_t$.

\subsection{Proof of Lemma~\ref{lemma2}} \label{appendixB.1}

\begin{proof}
Since $f$ is $L$-smooth, it follows from the descent lemma \citep{nesterov2018lectures} that, for $t \ge 0$,
\begin{equation} \label{B.1.1}
	\begin{aligned}
		&f({\hat x}^{(t+1)}) - f({\hat x}^{(t)}) \le \langle \nabla f({\hat x}^{(t)}), {\hat x}^{(t+1)} - {\hat x}^{(t)} \rangle + \frac{L}{2} \lVert {\hat x}^{(t+1)} - {\hat x}^{(t)} \rVert^2.
	\end{aligned}
\end{equation}
Substituting \eqref{B.0.3} into \eqref{B.1.1}, we obtain
\begin{equation} \label{B.1.2}
	\begin{aligned}
		&f({\hat x}^{(t+1)}) - f({\hat x}^{(t)}) \le - c \alpha \langle \nabla f({\hat x}^{(t)}), \bar m^{(t)} \rangle - \alpha \langle \nabla f({\hat x}^{(t)}), \pi_A^\top ({\mathbf y}^{(t+1)} - {\bar {\mathbf y}}^{(t+1)}) \rangle\\
		&\quad + c^2 \alpha^2 L \lVert\bar m ^{(t)} \rVert^2 + \alpha^2 L \lVert {\mathbf y}^{(t+1)} - {\bar {\mathbf y}}^{(t+1)} \rVert_F^2.
	\end{aligned}
\end{equation}

Using $\langle a,b \rangle = \frac{1}{2}(\lVert a \rVert^2 + \lVert b \rVert^2 - \lVert a - b \rVert^2)$ for any $a,b \in \mathbb R^p$, we have 
\begin{equation} \label{B.1.3}
	\begin{aligned}
		&-\langle \nabla f({\hat x}^{(t)}), \bar m^{(t)} \rangle\\
		&= -\frac{1}{2} \lVert\nabla f({\hat x}^{(t)}) \rVert^2 -\frac{1}{2} \lVert\bar m ^{(t)} \rVert^2 + \frac{1}{2}\lVert\nabla f({\hat x}^{(t)}) - \bar m ^{(t)} \rVert^2\\
		&\le -\frac{1}{2} \lVert\nabla f({\hat x}^{(t)}) \rVert^2 -\frac{1}{2} \lVert\bar m ^{(t)} \rVert^2 + \lVert\nabla f({\hat x}^{(t)}) - {\overline {\nabla f}} (\mathbf x^{(t)}) \rVert^2\\
		&\quad + \lVert {\bar m}^{(t)} - {\overline {\nabla f}} (\mathbf x^{(t)}) \rVert^2\\
		&\le -\frac{1}{2} \lVert\nabla f({\hat x}^{(t)}) \rVert^2 -\frac{1}{2} \lVert\bar m ^{(t)} \rVert^2 + \frac{L^2}{n} \lVert {\mathbf{x}}^{(t)} - \bar {\mathbf{x}}^{(t)} \rVert_F^2 + \lVert {\bar m}^{(t)} - {\overline {\nabla f}} (\mathbf x^{(t)}) \rVert^2,
	\end{aligned}
\end{equation}
where the last inequality is derived from the $L$-smoothness of $f_i$. Applying Young's inequality, we have
\begin{equation} \label{B.1.4}
	\begin{aligned}
		&\quad -\langle \nabla f({\hat x}^{(t)}), \pi_A^\top ({\mathbf y}^{(t+1)} - {\bar {\mathbf y}}^{(t+1)}) \rangle \le \frac{1}{c} \lVert {\mathbf y}^{(t+1)} - {\bar {\mathbf y}}^{(t+1)} \rVert^2_F + \frac{c}{4} \lVert\nabla f({\hat x}^{(t)}) \rVert^2.
	\end{aligned}
\end{equation}

Substituting \eqref{B.1.3} and \eqref{B.1.4} into \eqref{B.1.2} and using $\alpha \le \frac{1}{4cL}$ yields
\begin{equation} \label{B.1.5}
	\begin{aligned}
		&\lVert\nabla f({\hat x}^{(t)}) \rVert^2 \le \: \frac{4(f({\hat x}^{(t)}) - f({\hat x}^{(t+1)}))}{c \alpha} + \frac{4 L^2}{n} \lVert {\mathbf{x}}^{(t)} - \bar {\mathbf{x}}^{(t)} \rVert_F^2 + 4 \lVert {\bar m}^{(t)} - {\overline {\nabla f}} (\mathbf x^{(t)}) \rVert^2\\
		&\quad + \frac{5}{c^2} \lVert {\mathbf y}^{(t+1)} - {\bar {\mathbf y}}^{(t+1)} \rVert^2_F - \lVert\bar m ^{(t)} \rVert^2 .
	\end{aligned}
\end{equation}
Summing over $t$ from $0$ to $T-1$, using $f(x^{(0)}) - \inf_{x \in \mathbb R^p} f(x) \le \Delta_f$, and recalling that $x_i^{(0)} = x^{(0)}$ for all $i \in \mathcal N$, we have
\begin{equation} \label{B.1.6}
	\begin{aligned}
		&\sum_{t=0}^{T-1} \lVert\nabla f({\hat x}^{(t)}) \rVert^2 \le \: \frac{4 \Delta_f}{c \alpha} + \frac{4L^2}{n} \sum_{t=0}^{T-1} \lVert {\mathbf{x}}^{(t+1)} - \bar {\mathbf{x}}^{(t+1)} \rVert_F^2 + 4 \sum_{t=0}^{T-1} \lVert {\bar m}^{(t)} - {\overline {\nabla f}} (\mathbf x^{(t)}) \rVert^2\\
		&\quad + \frac{5}{c^2} \sum_{t=0}^{T-1} \lVert {\mathbf{y}}^{(t+1)} - \bar {\mathbf{y}}^{(t+1)} \rVert_F^2 - \sum_{t=0}^{T-1} \lVert {\bar m}^{(t)} \rVert^2.
	\end{aligned}
\end{equation}
Taking the expectation on both sides completes the proof. \qed
\end{proof}

\subsection{Proof of Lemma~\ref{lemma3}} \label{appendixB.2}

\begin{proof}
Using \eqref{B.0.2} recursively, we have, for $t\ge0$,
\begin{equation} \label{B.2.1}
	{\mathbf x}^{(t+1)} = A^{t+1} \mathbf x^{(0)} - \alpha \sum_{l=0}^{t} A^{t-l} \mathbf y^{(l+1)}.
\end{equation}
Left-multiplying \eqref{B.2.1} by $\mathds{1} \pi_A^\top$ yields
\begin{equation} \label{B.2.2}
	{\bar {\mathbf x}}^{(t+1)} = {\bar {\mathbf x}}^{(0)} - \alpha \sum_{l=0}^{t} \mathds{1} \pi_A^\top \mathbf y^{(l+1)}.
\end{equation}
Combining \eqref{B.2.1} and \eqref{B.2.2} with the initialization $x_i^{(0)} = x^{(0)}$ for all $i \in \mathcal N$ yields
\begin{equation} \label{B.2.3}
	\lVert {\mathbf x}^{(t+1)} - \bar {\mathbf x}^{(t+1)} \rVert_F^2 = \alpha^2 \Bigl\lVert\sum_{l=0}^{t} \bigl(A^{t-l} - \mathds 1 \pi_A^\top \bigr) \mathbf y^{(l+1)} \Bigr\rVert_F^2.
\end{equation}
Taking the sum of \eqref{B.2.3} over $t=0,\ldots,T-1$, we have
\begin{equation} \label{B.2.4}
	\begin{aligned}
		&\sum_{t=0}^{T-1} \lVert {\mathbf x}^{(t+1)} - \bar {\mathbf x}^{(t+1)} \rVert^2_F\\
		&= \: \alpha^2 \sum_{t=0}^{T-1} \Bigl\lVert\sum_{l=0}^{t} \bigl(A^{t-l} - \mathds 1 \pi_A^\top \bigr) \mathbf y^{(l+1)} \Bigr\rVert_F^2\\
		&\le \: s_A^2 \alpha^2 \sum_{t=0}^{T-1} \lVert\mathbf y^{(t+1)}\rVert_F^2\\
		&\le \: 2 s_A^2 \alpha^2 \sum_{t=0}^{T-1} \lVert {\mathbf{y}}^{(t+1)} - \bar {\mathbf{y}}^{(t+1)} \rVert_F^2 + 2 n^2 s_A^2 \alpha^2 \sum_{t=0}^{T-1} \lVert\bar m^{(t)}\rVert^2 ,
	\end{aligned}
\end{equation}
where the first inequality follows from Lemma~\ref{lemma11}. Taking the expectation on both sides completes the proof. \qed
\end{proof}

\subsection{Proof of Lemma~\ref{lemma4}} \label{appendixB.3}

\begin{proof}
Unrolling \eqref{B.0.1} and using $\mathbf y^{(0)}=\mathbb{0}$, we obtain, for $t\ge0$,
\begin{equation} \label{B.3.1}
	{\mathbf y}^{(t+1)} = \sum_{l=0}^{t} B^{t-l}(\mathbf m^{(l)} - \mathbf m^{(l-1)}).
\end{equation}
Left-multiplying \eqref{B.3.1} by $\pi_B \mathds{1}^\top$ yields
\begin{equation} \label{B.3.2}
	\bar {\mathbf y}^{(t+1)} = \sum_{l=0}^{t} \pi_B \mathds{1}^\top (\mathbf m^{(l)} - \mathbf m^{(l-1)}).
\end{equation}
Using \eqref{B.3.1}, \eqref{B.3.2}, and Lemma~\ref{lemma11}, we have
\begin{equation} \label{B.3.3}
	\sum_{t=0}^{T-1} \lVert {\mathbf y}^{(t+1)} - \bar {\mathbf y}^{(t+1)} \rVert_F^2 \le s_B^2 \sum_{t=0}^{T-1} \lVert\mathbf m^{(t)} - \mathbf m^{(t-1)}\rVert_F^2.
\end{equation}

Taking expectations and using $\mathbf m^{(-1)}=\mathbb{0}$, we obtain
\begin{equation} \label{B.3.4}
	\begin{aligned}
		&\sum_{t=0}^{T-1}\mathbb E\bigl[\lVert\mathbf y^{(t+1)}-\bar{\mathbf y}^{(t+1)}\rVert_F^2\bigr]\\
		&\le {}2s_B^2\mathbb E\bigl[\lVert\mathbf m^{(0)}-\nabla f(\mathbf x^{(0)})\rVert_F^2\bigr] +2s_B^2\lVert\nabla f(\mathbf x^{(0)})\rVert_F^2\\
		&\quad +s_B^2\sum_{t=1}^{T-1}\mathbb E\bigl[\lVert\mathbf m^{(t)}-\mathbf m^{(t-1)}\rVert_F^2\bigr]\\
		&\le {}\frac{2ns_B^2\sigma^2}{b_0} +4ns_B^2\Delta_fL +s_B^2\sum_{t=1}^{T-1}\mathbb E\bigl[\lVert\mathbf m^{(t)}-\mathbf m^{(t-1)}\rVert_F^2\bigr],
	\end{aligned}
\end{equation}
where the last inequality follows from Lemma~\ref{lemma12} together with the bound $\lVert\nabla f(\mathbf x^{(0)})\rVert_F^2\le2n\Delta_fL$, which follows from the $L$-smoothness of each $f_i$ and the definition of $\Delta_f$. Applying Lemma~\ref{lemma14} to \eqref{B.3.4} completes the proof. \qed
\end{proof}

\subsection{Proof of Lemma~\ref{lemma5}} \label{appendixB.4}

\begin{proof}
For any $i \in \mathcal N$ and $t \ge 1$, let $\delta_{1,i}^{(t)} := g_i(x_i^{(t)}, \xi_i^{(t)}) - \nabla f_i(x_i^{(t)})$ and $\delta_{2,i}^{(t)} := g_i(x_i^{(t-1)}, \xi_i^{(t)}) - \nabla f_i(x_i^{(t-1)})$. Using the update of $m_i^{(t)}$ in \eqref{eq3}, we obtain
\begin{equation} \label{B.4.1}
	\begin{aligned}
		&m_i^{(t)} - \nabla f_i(x_i^{(t)}) = \: \beta \bigl(m_i^{(t-1)} - \nabla f_i(x_i^{(t-1)})\bigr) + (1-\beta) \delta_{1,i}^{(t)} + \beta (\delta_{1,i}^{(t)} - \delta_{2,i}^{(t)}).
	\end{aligned}
\end{equation}
Averaging \eqref{B.4.1} over $i=1,\ldots,n$ yields
\begin{equation} \label{B.4.2}
	\begin{aligned}
		&\bar m^{(t)} - \overline {\nabla f}(\mathbf x^{(t)}) = \: \beta \bigl(\bar m^{(t-1)} - \overline {\nabla f}(\mathbf x^{(t-1)})\bigr) + \frac{1-\beta}{n} \sum_{i=1}^n \delta_{1,i}^{(t)} + \frac{\beta}{n} \sum_{i=1}^n (\delta_{1,i}^{(t)} - \delta_{2,i}^{(t)}).
	\end{aligned}
\end{equation}
From \eqref{B.4.2}, we can deduce that
\begin{equation} \label{B.4.3}
	\begin{aligned}
		&\mathbb E\bigl[\lVert\bar m^{(t)} - \overline {\nabla f}(\mathbf x^{(t)})\rVert^2 \bigm| \mathcal F_t\bigr]\\
		&= \:\beta^2 \lVert\bar m^{(t-1)} - \overline {\nabla f}(\mathbf x^{(t-1)})\rVert^2+ \mathbb E \Bigl[ \Bigl\lVert\frac{(1-\beta)}{n} \sum_{i=1}^n \delta_{1,i}^{(t)} + \frac{\beta}{n} \sum_{i=1}^n (\delta_{1,i}^{(t)} - \delta_{2,i}^{(t)}) \Bigr\rVert^2 \Bigm| \mathcal F_t \Bigr]\\
		&\le \:\beta^2 \lVert\bar m^{(t-1)} - \overline {\nabla f}(\mathbf x^{(t-1)})\rVert^2 + \frac{2(1-\beta)^2}{n^2}\mathbb E \Bigl[ \Bigl\lVert\sum_{i=1}^n \delta_{1,i}^{(t)} \Bigr\rVert^2 \Bigm| \mathcal F_t \Bigr]\\
		&\quad + \frac{2 \beta^2}{n^2}\mathbb E \Bigl[ \Bigl\lVert\sum_{i=1}^n (\delta_{1,i}^{(t)} - \delta_{2,i}^{(t)}) \Bigr\rVert^2 \Bigm| \mathcal F_t \Bigr],
	\end{aligned}
\end{equation}
where the equality follows from the unbiasedness of the stochastic gradients and the $\mathcal F_t$-measurability of $\bar m^{(t-1)}-\overline{\nabla f}(\mathbf x^{(t-1)})$.

For the second and third terms on the right-hand side of \eqref{B.4.3}, the unbiasedness of the stochastic gradients and the independence of the samples give
\begin{equation} \label{B.4.5}
	\begin{aligned}
		&\mathbb E\Bigl[ \Bigl\lVert\sum_{i=1}^n\delta_{1,i}^{(t)}\Bigr\rVert^2 \Bigm|\mathcal F_t \Bigr] ={}\sum_{i=1}^n \mathbb E\bigl[ \lVert\delta_{1,i}^{(t)}\rVert^2 \bigm|\mathcal F_t \bigr] \le n\sigma^2,
	\end{aligned}
\end{equation}
and
\begin{equation} \label{B.4.7}
	\begin{aligned}
		&\mathbb E\Bigl[ \Bigl\lVert\sum_{i=1}^{n} \bigl(\delta_{1,i}^{(t)}-\delta_{2,i}^{(t)}\bigr) \Bigr\rVert^{2} \Bigm|\mathcal F_t \Bigr]\\
		&= {}\sum_{i=1}^{n} \mathbb E\bigl[ \lVert\delta_{1,i}^{(t)}-\delta_{2,i}^{(t)}\rVert^{2} \bigm|\mathcal F_t \bigr]\\
		&= {}\sum_{i=1}^{n} \mathbb E\bigl[ \lVert g_i(x_i^{(t)},\xi_i^{(t)}) -g_i(x_i^{(t-1)},\xi_i^{(t)})\rVert^{2} \bigm|\mathcal F_t \bigr] -\sum_{i=1}^{n} \lVert\nabla f_i(x_i^{(t)}) -\nabla f_i(x_i^{(t-1)})\rVert^{2}\\
		&\le {}L^{2} \lVert\mathbf x^{(t)}-\mathbf x^{(t-1)}\rVert_F^{2}.
	\end{aligned}
\end{equation}
Here, the second equality in \eqref{B.4.7} follows by expanding the squared norm and using the unbiasedness of the stochastic gradients. The final inequalities in \eqref{B.4.5} and \eqref{B.4.7} follow from the variance bound and mean-squared smoothness, respectively.

Substituting \eqref{B.4.5} and \eqref{B.4.7} into \eqref{B.4.3} and taking the total expectation yields
\begin{equation} \label{B.4.8}
	\begin{aligned}
		&\mathbb E\bigl[\lVert\bar m^{(t)} - \overline {\nabla f}(\mathbf x^{(t)})\rVert^2 \bigr]\\
		&\le \:\beta^2 \mathbb E\bigl[\lVert\bar m^{(t-1)} - \overline {\nabla f}(\mathbf x^{(t-1)})\rVert^2\bigr] + \frac{2(1-\beta)^2 \sigma^2}{n}+ \frac{2 \beta^2 L^2}{n^2} \mathbb E \bigl[\lVert {\mathbf x}^{(t)} - {\mathbf x}^{(t-1)} \rVert_F^2\bigr]\\
		&\le \: \beta^{2t} \mathbb E\bigl[\lVert\bar m^{(0)} - \overline {\nabla f}(\mathbf x^{(0)})\rVert^2\bigr] + \frac{2(1-\beta)^2 \sigma^2}{n} \sum_{l=0}^{t-1} \beta^{2(t-1-l)}\\
		&\quad + \frac{2 \beta^2 L^2}{n^2} \sum_{l=0}^{t-1} \beta^{2(t-1-l)} \mathbb E \bigl[\lVert {\mathbf x}^{(l+1)} - {\mathbf x}^{(l)} \rVert_F^2\bigr],
	\end{aligned}
\end{equation}
where the last inequality is obtained by unrolling the recurrence down to $t=0$. Summing over $t=0,\ldots,T-1$, we have
\begin{equation} \label{B.4.9}
	\begin{aligned}
		&\sum_{t=0}^{T-1}\mathbb E\bigl[\lVert\bar m^{(t)} - \overline {\nabla f}(\mathbf x^{(t)})\rVert^2 \bigr]\\
		&\le \: \sum_{t=0}^{T-1} \beta^{2t} \mathbb E\bigl[\lVert\bar m^{(0)} - \overline {\nabla f}(\mathbf x^{(0)})\rVert^2\bigr] + \frac{2(1-\beta)^2 \sigma^2}{n} \sum_{t=1}^{T-1} \sum_{l=0}^{t-1} \beta^{2(t-1-l)}\\
		&\quad + \frac{2 \beta^2 L^2}{n^2} \sum_{t=1}^{T-1} \sum_{l=0}^{t-1} \beta^{2(t-1-l)}\mathbb E \bigl[\lVert {\mathbf x}^{(l+1)} - {\mathbf x}^{(l)} \rVert_F^2\bigr]\\
		&\le \: \frac{1}{1 - \beta^2}\mathbb E\bigl[\lVert\bar m^{(0)} - \overline {\nabla f}(\mathbf x^{(0)})\rVert^2\bigr] + \frac{2T(1-\beta)^2 \sigma^2}{n(1-\beta^2)}\\
		&\quad + \frac{2 \beta^2 L^2}{n^2(1 - \beta^2)} \sum_{t=0}^{T-1} \mathbb E \bigl[\lVert {\mathbf x}^{(t+1)} - {\mathbf x}^{(t)} \rVert_F^2\bigr],
	\end{aligned}
\end{equation}
where the last inequality is obtained by interchanging the order of summation and using the geometric-series bound. Applying Lemma~\ref{lemma12} to \eqref{B.4.9} and using $\beta \in [0, 1)$ yields the desired result. \qed
\end{proof}

\subsection{Proof of Lemma~\ref{lemma6}} \label{appendixB.5}

\begin{proof}
From \eqref{B.0.2}, it follows that, for $t\ge1$,
\begin{equation} \label{B.5.1}
	\mathbf x^{(t+1)}-\mathbf x^{(t)}
	=A(\mathbf x^{(t)}-\mathbf x^{(t-1)})-\alpha(\mathbf y^{(t+1)}-\mathbf y^{(t)}).
\end{equation}
The common initialization gives $A\mathbf x^{(0)}=\mathbf x^{(0)}$ and $\mathbf x^{(1)}-\mathbf x^{(0)}=-\alpha\mathbf y^{(1)}$. Unrolling \eqref{B.5.1} and using $\mathbf y^{(0)}=\mathbb{0}$, we obtain, for $t\ge0$,
\begin{equation} \label{B.5.2}
	\begin{aligned}
		&\mathbf x^{(t+1)}-\mathbf x^{(t)}\\
		&= {} -\alpha\sum_{l=0}^{t}A^{t-l}\bigl(\mathbf y^{(l+1)}-\mathbf y^{(l)}\bigr)\\
		&= {}-c\alpha\mathds{1}\bar m^{(t)} -\alpha\mathds{1}\pi_A^\top\bigl(\mathbf y^{(t+1)}-\bar{\mathbf y}^{(t+1)}\bigr) -\alpha\sum_{l=0}^{t}(A^{t-l}-\mathds{1}\pi_A^\top)\bigl(\mathbf y^{(l+1)}-\mathbf y^{(l)}\bigr),
	\end{aligned}
\end{equation}
where the last equality follows from $\pi_A^\top\bar{\mathbf y}^{(t+1)}=\pi_A^\top\pi_B \mathds{1}^\top \mathbf m^{(t)}=c\bar m^{(t)}$.

By $\lVert\mathds1\pi_A^\top\rVert_2^2\le n$ and Lemma~\ref{lemma11}, it follows from \eqref{B.5.2} that
\begin{equation} \label{B.5.3}
	\begin{aligned}
		&\sum_{t=0}^{T-1}\lVert\mathbf x^{(t+1)}-\mathbf x^{(t)}\rVert_F^2 \le{}3nc^2\alpha^2\sum_{t=0}^{T-1}\lVert\bar m^{(t)}\rVert^2+3n\alpha^2\sum_{t=0}^{T-1}\lVert\mathbf y^{(t+1)}-\bar{\mathbf y}^{(t+1)}\rVert_F^2\\
		&\quad +3s_A^2\alpha^2\sum_{t=0}^{T-1}\lVert\mathbf y^{(t+1)}-\mathbf y^{(t)}\rVert_F^2.
	\end{aligned}
\end{equation}

For the second term on the right-hand side of \eqref{B.5.3}, we use \eqref{B.3.3}. For the third term, note that
\begin{equation} \label{B.5.6}
	\begin{aligned}
		&\mathbf y^{(t+1)}-\mathbf y^{(t)} ={}\pi_B\mathds1^\top\bigl(\mathbf m^{(t)}-\mathbf m^{(t-1)}\bigr) +\bigl(\mathbf y^{(t+1)}-\bar{\mathbf y}^{(t+1)}\bigr)-\bigl(\mathbf y^{(t)}-\bar{\mathbf y}^{(t)}\bigr).
	\end{aligned}
\end{equation}
Using $\lVert\pi_B\mathds1^\top\rVert_2^2\le n$ and $\mathbf y^{(0)}=\bar{\mathbf y}^{(0)}=\mathbb{0}$, we have
\begin{equation} \label{B.5.7}
	\begin{aligned}
		&\sum_{t=0}^{T-1}\lVert\mathbf y^{(t+1)}-\mathbf y^{(t)}\rVert_F^2 \le{}3n\sum_{t=0}^{T-1}\lVert\mathbf m^{(t)}-\mathbf m^{(t-1)}\rVert_F^2+6\sum_{t=0}^{T-1}\lVert\mathbf y^{(t+1)}-\bar{\mathbf y}^{(t+1)}\rVert_F^2.
	\end{aligned}
\end{equation}
Substituting \eqref{B.3.3} and \eqref{B.5.7} into \eqref{B.5.3} and using $9ns_A^2+3ns_B^2+18s_A^2s_B^2\le18h^2$ gives
\begin{equation} \label{B.5.8}
	\begin{aligned}
		&\sum_{t=0}^{T-1}\lVert\mathbf x^{(t+1)}-\mathbf x^{(t)}\rVert_F^2 \le{}3nc^2\alpha^2\sum_{t=0}^{T-1}\lVert\bar m^{(t)}\rVert^2+18\alpha^2h^2\sum_{t=0}^{T-1}\lVert\mathbf m^{(t)}-\mathbf m^{(t-1)}\rVert_F^2.
	\end{aligned}
\end{equation}

Taking expectations and using $\mathbf m^{(-1)}=\mathbb{0}$, we obtain
\begin{equation} \label{B.5.9}
	\begin{aligned}
		&\sum_{t=0}^{T-1}\mathbb E\bigl[\lVert\mathbf x^{(t+1)}-\mathbf x^{(t)}\rVert_F^2\bigr]\\
		&\le {}3nc^2\alpha^2\sum_{t=0}^{T-1}\mathbb E\bigl[\lVert\bar m^{(t)}\rVert^2\bigr] +18\alpha^2h^2\sum_{t=1}^{T-1}\mathbb E\bigl[\lVert\mathbf m^{(t)}-\mathbf m^{(t-1)}\rVert_F^2\bigr]\\
		&\quad + 36\alpha^2h^2\mathbb E \bigl[\lVert\mathbf{m}^{(0)} - {\nabla f}(\mathbf x^{(0)}) \rVert_F^2 \bigr] + 36\alpha^2h^2\lVert {\nabla f}(\mathbf x^{(0)}) \rVert_F^2\\
		&\le {}3nc^2\alpha^2\sum_{t=0}^{T-1}\mathbb E\bigl[\lVert\bar m^{(t)}\rVert^2\bigr] +162\alpha^2h^2L^2\sum_{t=0}^{T-1}\mathbb E\bigl[\lVert\mathbf x^{(t+1)}-\mathbf x^{(t)}\rVert_F^2\bigr]\\
		&\quad + 162nT\alpha^2h^2(1-\beta)^2\sigma^2 +\frac{90n\alpha^2h^2\sigma^2}{b_0} + 72n\alpha^2h^2\Delta_fL,
	\end{aligned}
\end{equation}
where the last inequality follows from Lemmas~\ref{lemma12} and~\ref{lemma14} together with the bound $\lVert\nabla f(\mathbf x^{(0)})\rVert_F^2\le 2n\Delta_fL$. Finally, $\alpha\le\frac{1}{20hL}$ implies $1-162\alpha^2h^2L^2\ge\frac12$. Rearranging \eqref{B.5.9} and using this bound completes the proof. \qed
\end{proof}

\subsection{Proof of Theorem~\ref{theorem1}} \label{appendixB.6}

\begin{proof}
Substituting Lemma~\ref{lemma3} into Lemma~\ref{lemma2} yields
\begin{equation} \label{B.6.1}
	\begin{aligned}
		&\sum_{t=0}^{T-1}\mathbb E\bigl[\lVert\nabla f(\hat x^{(t)})\rVert^2\bigr] \le{}\frac{4\Delta_f}{c\alpha} -\frac{3}{4}\sum_{t=0}^{T-1}\mathbb E\bigl[\lVert\bar m^{(t)}\rVert^2\bigr] +\frac{6}{c^2}\sum_{t=0}^{T-1}\mathbb E\bigl[\lVert\mathbf y^{(t+1)}-\bar{\mathbf y}^{(t+1)}\rVert_F^2\bigr]\\
		&\quad +4\sum_{t=0}^{T-1}\mathbb E\bigl[\lVert\bar m^{(t)}-\overline{\nabla f}(\mathbf x^{(t)})\rVert^2\bigr],
	\end{aligned}
\end{equation}
where we used $8ns_A^2\alpha^2L^2-1\le-\frac{3}{4}$ and $\frac{8s_A^2\alpha^2L^2}{n}+\frac{5}{c^2}\le\frac{6}{c^2}$, which follow from $\alpha\le\frac{1}{20hL}$, $h^2\ge ns_A^2$, and $c\le n$.

Substituting the bounds in Lemmas~\ref{lemma4} and~\ref{lemma5} into \eqref{B.6.1} yields
\begin{equation} \label{B.6.2}
	\begin{aligned}
		&\sum_{t=0}^{T-1} \mathbb E\bigl[\lVert\nabla f(\hat x^{(t)})\rVert^2\bigr] \le{} \frac{4\Delta_f}{c\alpha} -\frac{3}{4}\sum_{t=0}^{T-1} \mathbb E\bigl[\lVert\bar m^{(t)}\rVert^2\bigr]\\
		&\quad +\Bigl( \frac{54s_B^2L^2}{c^2} +\frac{8L^2}{n^2(1-\beta)} \Bigr) \sum_{t=0}^{T-1} \mathbb E\bigl[ \lVert\mathbf x^{(t+1)}-\mathbf x^{(t)}\rVert_F^2 \bigr]\\
		&\quad +\frac{4\sigma^2}{nb_0(1-\beta)} +\frac{8T(1-\beta)\sigma^2}{n} +\frac{54nTs_B^2(1-\beta)^2\sigma^2}{c^2} +\frac{30ns_B^2\sigma^2}{c^2b_0} +\frac{24ns_B^2\Delta_fL}{c^2}.
	\end{aligned}
\end{equation}

Applying Lemma~\ref{lemma6} to \eqref{B.6.2} gives
\begin{equation} \label{B.6.3}
	\begin{aligned}
		&\sum_{t=0}^{T-1} \mathbb E\bigl[\lVert\nabla f(\hat x^{(t)})\rVert^2\bigr] \le{} \frac{4\Delta_f}{c\alpha} +\frac{4\sigma^2}{nb_0(1-\beta)} +\frac{8T(1-\beta)\sigma^2}{n}\\
		&\quad +\frac{80Th^2(1-\beta)^2\sigma^2}{c^2} +\frac{45h^2\sigma^2}{c^2b_0} +\frac{36h^2\Delta_fL}{c^2},
	\end{aligned}
\end{equation}
where we used $ns_B^2\le h^2$, $\alpha^2h^2L^2\le\frac{1}{900}$, and $\frac{c^2\alpha^2L^2}{n(1-\beta)}\le\frac{1}{400}$. In particular, the coefficient of $\sum_{t=0}^{T-1}\mathbb E\bigl[\lVert\bar m^{(t)}\rVert^2\bigr]$ satisfies
$324ns_B^2\alpha^2L^2
+\frac{48c^2\alpha^2L^2}{n(1-\beta)}-\frac{3}{4}\le0
$.
Dividing \eqref{B.6.3} by $T$ completes the proof. \qed
\end{proof}

\section{Proofs for LASPP}

For Algorithm~\ref{alg2}, let $\mathcal F_{t,k}$ denote the $\sigma$-algebra generated by all samples drawn before local step~$k$ of communication round~$t$. Set \[\mathcal F_{1,0}:=\sigma\bigl(\xi_{i,r}^{(0)}:i\in\mathcal N,\ 1\le r\le b_0\bigr).\] For $t\ge1$, recursively define \[\begin{aligned}\mathcal F_{t,k+1}&:=\sigma\bigl(\mathcal F_{t,k},\xi_i^{(t,k)}:i\in\mathcal N\bigr),\quad k=0,\ldots,K-1,\\ \mathcal F_{t+1,0}&:=\mathcal F_{t,K}.\end{aligned}\]
The update rules of Algorithm~\ref{alg2} imply that $\mathbf x^{(t)}$, $\mathbf y^{(t)}$, and $\mathbf m^{(t-1)}$ are $\mathcal F_{t,0}$-measurable. Moreover, $\mathbf x^{(t,k)}$ is $\mathcal F_{t,k}$-measurable for $k=0,\ldots,K-1$. By Assumption~\ref{assumption1}, the samples $\{\xi_i^{(t,k)}\}_{i=1}^n$ are mutually independent and independent of $\mathcal F_{t,k}$.

\subsection{Proof of Lemma~\ref{lemma7}} \label{appendixC.1}

\begin{proof}
By the same argument as in \eqref{B.3.1}--\eqref{B.3.4}, we obtain
\begin{equation} \label{C.1.1}
		\sum_{t=0}^{T-1}\mathbb E\bigl[\lVert\mathbf y^{(t+1)}-\bar{\mathbf y}^{(t+1)}\rVert_F^2\bigr]
		\le\frac{2ns_B^2\sigma^2}{b_0}+4ns_B^2\Delta_fL
		+s_B^2\sum_{t=1}^{T-1}\mathbb E\bigl[\lVert\mathbf m^{(t)}-\mathbf m^{(t-1)}\rVert_F^2\bigr].
\end{equation}
Applying Lemma~\ref{lemma16} to \eqref{C.1.1} completes the proof. \qed
\end{proof}

\subsection{Proof of Lemma~\ref{lemma8}} \label{appendixC.2}

\begin{proof}
For any $i\in\mathcal N$, $t\ge1$, and $k=0,\ldots,K-1$, let $\delta_{1,i}^{(t,k)} := g_i(x_i^{(t)}, \xi_i^{(t,k)}) - \nabla f_i(x_i^{(t)})$, $\delta_{2,i}^{(t,k)} := g_i(x_i^{(t-1)}, \xi_i^{(t,k)}) - \nabla f_i(x_i^{(t-1)})$, and $\Delta_i^{(t,k)} := g_i(x_i^{(t,k)}, \xi_i^{(t,k)}) - g_i(x_i^{(t)}, \xi_i^{(t,k)})$. From \eqref{eq5}, it follows that
\begin{equation} \label{C.2.1}
	\begin{aligned}
		&m_i^{(t)} - \nabla f_i(x_i^{(t)}) = \: \beta \bigl(m_i^{(t-1)} - \nabla f_i(x_i^{(t-1)})\bigr) + \frac{1-\beta}{K} \sum_{k=0}^{K-1} \delta_{1,i}^{(t,k)}\\
		&\quad + \frac{\beta}{K} \sum_{k=0}^{K-1}(\delta_{1,i}^{(t,k)} - \delta_{2,i}^{(t,k)}) + \frac{1}{K}\sum_{k=0}^{K-1}\Delta_i^{(t,k)}.
	\end{aligned}
\end{equation}
Averaging both sides of \eqref{C.2.1} over $i=1,\ldots,n$ gives
\begin{equation} \label{C.2.2}
	\begin{aligned}
		&\bar m^{(t)} - \overline {\nabla f}(\mathbf x^{(t)}) = \: \beta \bigl(\bar m^{(t-1)} - \overline {\nabla f}(\mathbf x^{(t-1)})\bigr) + \frac{1-\beta}{nK} \sum_{i=1}^n \sum_{k=0}^{K-1} \delta_{1,i}^{(t,k)}\\
		&\quad + \frac{\beta}{nK} \sum_{i=1}^n \sum_{k=0}^{K-1}(\delta_{1,i}^{(t,k)} - \delta_{2,i}^{(t,k)}) + \frac{1}{nK}\sum_{i=1}^n \sum_{k=0}^{K-1}\Delta_i^{(t,k)}.
	\end{aligned}
\end{equation}
Using the unbiasedness of the stochastic gradients and the $\mathcal F_{t,0}$-measurability of $\bar m^{(t-1)}-\overline{\nabla f}(\mathbf x^{(t-1)})$, we obtain from \eqref{C.2.2}
\begin{equation} \label{C.2.3}
	\begin{aligned}
		&\mathbb E\bigl[\lVert\bar m^{(t)} - \overline {\nabla f}(\mathbf x^{(t)})\rVert^2 \bigm| \mathcal F_{t,0}\bigr] \le \: \beta^2 \lVert\bar m^{(t-1)} - \overline {\nabla f}(\mathbf x^{(t-1)})\rVert^2\\
		&+ \frac{3(1-\beta)^2}{n^2 K^2} \mathbb E \Bigl[ \Bigl\lVert\sum_{i=1}^n \sum_{k=0}^{K-1} \delta_{1,i}^{(t,k)} \Bigr\rVert^2 \Bigm| \mathcal F_{t,0} \Bigr]\\
		&+ \frac{3 \beta^2}{n^2 K^2} \mathbb E \Bigl[ \Bigl\lVert\sum_{i=1}^n \sum_{k=0}^{K-1}(\delta_{1,i}^{(t,k)} - \delta_{2,i}^{(t,k)}) \Bigr\rVert^2 \Bigm| \mathcal F_{t,0} \Bigr] + \frac{3}{n^2 K^2}\mathbb E \Bigl[ \Bigl\lVert\sum_{i=1}^n \sum_{k=0}^{K-1} \Delta_i^{(t,k)} \Bigr\rVert^2 \Bigm| \mathcal F_{t,0} \Bigr]\\
		&+ 2\Bigl\langle \beta \bigl(\bar m^{(t-1)} \!\!-\! \overline {\nabla f}(\mathbf x^{(t-1)})\bigr), \frac{1}{nK} \mathbb E \Bigl[\sum_{i=1}^n \! \sum_{k=0}^{K-1}\! \Delta_i^{(t,k)} \Bigm| \mathcal F_{t,0} \Bigr] \Bigr\rangle.
	\end{aligned}
\end{equation}

For the second and third terms on the right-hand side of \eqref{C.2.3}, the unbiasedness of the stochastic gradients and the independence of the samples give
\begin{equation} \label{C.2.5}
	\begin{aligned}
		&\mathbb E \Bigl[ \Bigl\lVert\sum_{i=1}^n \sum_{k=0}^{K-1} \delta_{1,i}^{(t,k)} \Bigr\rVert^2 \Bigm| \mathcal F_{t,0} \Bigr] = \sum_{i=1}^n \sum_{k=0}^{K-1} \mathbb E \bigl[ \lVert\delta_{1,i}^{(t,k)} \rVert^2 \bigm| \mathcal F_{t,0} \bigr] \le nK \sigma^2,
	\end{aligned}
\end{equation}
and
\begin{equation} \label{C.2.6}
	\begin{aligned}
		&\mathbb E \Bigl[ \Bigl\lVert\sum_{i=1}^n \sum_{k=0}^{K-1}(\delta_{1,i}^{(t,k)} - \delta_{2,i}^{(t,k)}) \Bigr\rVert^2 \Bigm| \mathcal F_{t,0} \Bigr]\\
		&= \sum_{i=1}^n \sum_{k=0}^{K-1} \mathbb E \bigl[ \lVert\delta_{1,i}^{(t,k)} - \delta_{2,i}^{(t,k)} \rVert^2 \bigm| \mathcal F_{t,0} \bigr]\\
		&= \sum_{i=1}^n \sum_{k=0}^{K-1} \mathbb E \bigl[\bigl\lVert g_i(x_i^{(t)}, \xi_i^{(t,k)}) - g_i(x_i^{(t-1)}, \xi_i^{(t,k)})\bigr\rVert^2 \bigm|\mathcal F_{t,0}\bigr]\\
		&\quad - K\sum_{i=1}^n\bigl\lVert\nabla f_i(x_i^{(t)}) - \nabla f_i(x_i^{(t-1)}) \bigr\rVert^2\\
		&\le K L^2 \lVert {\mathbf x}^{(t)} - {\mathbf x}^{(t-1)} \rVert_F^2.
	\end{aligned}
\end{equation}
The final inequalities in \eqref{C.2.5} and \eqref{C.2.6} follow from the variance bound and mean-squared smoothness, respectively.

For the fourth term in \eqref{C.2.3}, we have
\begin{equation} \label{C.2.7}
	\begin{aligned}
		&\mathbb E \Bigl[ \Bigl\lVert\sum_{i=1}^n \! \sum_{k=0}^{K-1} \Delta_i^{(t,k)} \Bigr\rVert^2 \Bigm| \mathcal F_{t,0} \Bigr]\\
		&\le \: nK \sum_{i=1}^n \! \sum_{k=0}^{K-1} \mathbb E \bigl[ \bigl\lVert g_i(x_i^{(t,k)}\!, \xi_i^{(t,k)}) \! - g_i(x_i^{(t)}\!, \xi_i^{(t,k)}) \bigr\rVert^2 \bigm| \mathcal F_{t,0} \bigr]\\
		&\le \: nK L^2 \sum_{k=0}^{K-1} \mathbb E \bigl[ \lVert {\mathbf x}^{(t,k)} - {\mathbf x}^{(t)} \rVert_F^2 \bigm|\mathcal F_{t,0}\bigr],
	\end{aligned}
\end{equation}
where the last inequality is due to
\begin{equation} \label{C.2.8}
	\begin{aligned}
		&\mathbb E \bigl[ \bigl\lVert g_i(x_i^{(t,k)}\!, \xi_i^{(t,k)}) \! - g_i(x_i^{(t)}\!, \xi_i^{(t,k)}) \bigr\rVert^2 \bigm| \mathcal F_{t,0} \bigr]\\
		&= \: \mathbb E \Bigl[ \mathbb E \bigl[ \bigl\lVert g_i(x_i^{(t,k)}\!, \xi_i^{(t,k)}) \! - g_i(x_i^{(t)}\!, \xi_i^{(t,k)}) \bigr\rVert^2 \bigm| \mathcal F_{t,k} \bigr] \Bigm| \mathcal F_{t,0} \Bigr]\\
		&\le \: L^2 \mathbb E \bigl[ \lVert x_i^{(t,k)} - x_i^{(t)} \rVert^2 \bigm|\mathcal F_{t,0}\bigr],
	\end{aligned}
\end{equation}
where the equality follows from the tower property of conditional expectation and the last inequality follows from mean-squared smoothness.

For the last term in \eqref{C.2.3}, Young's inequality and Jensen's inequality give
\begin{equation} \label{C.2.9}
	\begin{aligned}
		&2\Bigl\langle \beta \bigl(\bar m^{(t-1)} \!\!-\! \overline {\nabla f}(\mathbf x^{(t-1)})\bigr), \frac{1}{nK} \mathbb E \Bigl[\sum_{i=1}^n \! \sum_{k=0}^{K-1}\! \Delta_i^{(t,k)} \Bigm| \mathcal F_{t,0} \Bigr] \Bigr\rangle\\
		&\le \: (1-\beta) \beta^2 \lVert\bar m^{(t-1)} - \overline {\nabla f}(\mathbf x^{(t-1)})\rVert^2 + \frac{1}{n^2 K^2 (1-\beta)} \mathbb E \Bigl[ \Bigl\lVert\sum_{i=1}^n \! \sum_{k=0}^{K-1} \Delta_i^{(t,k)} \Bigr\rVert^2 \Bigm| \mathcal F_{t,0} \Bigr]\\
		&\le \: (1-\beta) \beta^2 \lVert\bar m^{(t-1)} - \overline {\nabla f}(\mathbf x^{(t-1)})\rVert^2 + \frac{L^2}{n K (1-\beta)} \sum_{k=0}^{K-1} \mathbb E \bigl[ \lVert {\mathbf x}^{(t,k)} - {\mathbf x}^{(t)} \rVert_F^2 \bigm|\mathcal F_{t,0}\bigr],
	\end{aligned}
\end{equation}
where the last inequality follows from \eqref{C.2.7}.

Substituting \eqref{C.2.5}, \eqref{C.2.6}, \eqref{C.2.7}, and \eqref{C.2.9} into \eqref{C.2.3}, and then using $\beta\in[0,1)$ and taking the total expectation, we obtain
\begin{equation} \label{C.2.10}
	\begin{aligned}
		&\mathbb E\bigl[\lVert\bar m^{(t)} - \overline {\nabla f}(\mathbf x^{(t)})\rVert^2\bigr]\\
		&\le \: \beta \mathbb E\bigl[\lVert\bar m^{(t-1)} - \overline {\nabla f}(\mathbf x^{(t-1)})\rVert^2\bigr] + \frac{4 L^2}{n K (1-\beta)} \sum_{k=0}^{K-1} \mathbb E \bigl[ \lVert {\mathbf x}^{(t,k)} - {\mathbf x}^{(t)} \rVert_F^2 \bigr]\\
		&\quad + \frac{3 \beta^2 L^2}{n^2 K} \mathbb E \bigl[ \lVert {\mathbf x}^{(t)} - {\mathbf x}^{(t-1)} \rVert_F^2\bigr] + \frac{3(1-\beta)^2 \sigma^2}{n K}\\
		&\le \: \beta^t \mathbb E\bigl[\lVert\bar m^{(0)} - \overline {\nabla f}(\mathbf x^{(0)})\rVert^2\bigr] + \frac{4 L^2}{n K (1-\beta)} \sum_{l=1}^t \sum_{k=0}^{K-1} \beta^{t-l} \mathbb E \bigl[ \lVert {\mathbf x}^{(l,k)} - {\mathbf x}^{(l)} \rVert_F^2 \bigr]\\
		&\quad + \frac{3 \beta^2 L^2}{n^2 K} \sum_{l=0}^{t-1} \beta^{t-1-l} \mathbb E \bigl[ \lVert {\mathbf x}^{(l+1)} - {\mathbf x}^{(l)} \rVert_F^2\bigr] + \frac{3(1-\beta)^2 \sigma^2}{n K} \sum_{l=0}^{t-1} \beta^{t-1-l},
	\end{aligned}
\end{equation}
where the first inequality uses $\beta^2(2-\beta)\le\beta$ and $3+\frac{1}{1-\beta}\le\frac{4}{1-\beta}$, and the second follows by unrolling the recurrence to $t=0$. 

Summing over $t=0,\ldots,T-1$ yields
\begin{equation} \label{C.2.11}
	\begin{aligned}
		&\sum_{t=0}^{T-1}\mathbb E\bigl[\lVert\bar m^{(t)} - \overline {\nabla f}(\mathbf x^{(t)})\rVert^2\bigr]\\
		&\le \: \sum_{t=0}^{T-1}\beta^t \mathbb E\bigl[\lVert\bar m^{(0)} - \overline {\nabla f}(\mathbf x^{(0)})\rVert^2\bigr]+ \frac{4 L^2}{n K (1-\beta)} \sum_{t=1}^{T-1}\sum_{l=1}^t \sum_{k=0}^{K-1} \beta^{t-l} \mathbb E \bigl[ \lVert {\mathbf x}^{(l,k)} - {\mathbf x}^{(l)} \rVert_F^2 \bigr]\\
		&\quad + \frac{3 \beta^2 L^2}{n^2 K} \sum_{t=1}^{T-1}\sum_{l=0}^{t-1} \beta^{t-1-l} \mathbb E \bigl[ \lVert {\mathbf x}^{(l+1)} - {\mathbf x}^{(l)} \rVert_F^2\bigr] + \frac{3(1-\beta)^2 \sigma^2}{n K} \sum_{t=1}^{T-1}\sum_{l=0}^{t-1} \beta^{t-1-l}\\
		&\le \: \frac{1}{1-\beta} \mathbb E\bigl[\lVert\bar m^{(0)} - \overline {\nabla f}(\mathbf x^{(0)})\rVert^2\bigr] + \frac{4 L^2}{n K (1-\beta)^2} \sum_{t=1}^{T-1} \sum_{k=0}^{K-1} \mathbb E \bigl[ \lVert {\mathbf x}^{(t,k)} - {\mathbf x}^{(t)} \rVert_F^2 \bigr]\\
		&\quad + \frac{3 \beta^2 L^2}{n^2 K (1-\beta)} \sum_{t=0}^{T-1} \mathbb E \bigl[ \lVert {\mathbf x}^{(t+1)} - {\mathbf x}^{(t)} \rVert_F^2\bigr] + \frac{3T(1-\beta) \sigma^2}{n K},
	\end{aligned}
\end{equation}
where the last inequality is obtained by interchanging the order of summation
and using the geometric-series bound. Applying Lemma~\ref{lemma12} to \eqref{C.2.11} and using $\beta^2\le1$ completes the proof. \qed
\end{proof}

\subsection{Proof of Lemma~\ref{lemma9}} \label{appendixC.3}

\begin{proof}
Applying the argument leading to \eqref{B.5.8} to Algorithm~\ref{alg2}, we obtain
\begin{equation} \label{C.3.1}
	\begin{aligned}
		&\sum_{t=0}^{T-1}\lVert\mathbf x^{(t+1)}-\mathbf x^{(t)}\rVert_F^2 \le{}3nc^2\alpha^2\sum_{t=0}^{T-1}\lVert\bar m^{(t)}\rVert^2+18\alpha^2h^2\sum_{t=0}^{T-1}\lVert\mathbf m^{(t)}-\mathbf m^{(t-1)}\rVert_F^2.
	\end{aligned}
\end{equation}

Taking expectations and using $\mathbf m^{(-1)}=\mathbb{0}$, we obtain
\begin{equation} \label{C.3.2}
	\begin{aligned}
		&\sum_{t=0}^{T-1}\mathbb E\bigl[\lVert\mathbf x^{(t+1)}-\mathbf x^{(t)}\rVert_F^2\bigr]\\
		&\le {}3nc^2\alpha^2\sum_{t=0}^{T-1}\mathbb E\bigl[\lVert\bar m^{(t)}\rVert^2\bigr] +18\alpha^2h^2\sum_{t=1}^{T-1}\mathbb E\bigl[\lVert\mathbf m^{(t)}-\mathbf m^{(t-1)}\rVert_F^2\bigr]\\
		&\quad +36\alpha^2h^2\mathbb E\bigl[\lVert\mathbf m^{(0)}-\nabla f(\mathbf x^{(0)})\rVert_F^2\bigr] +36\alpha^2h^2\lVert\nabla f(\mathbf x^{(0)})\rVert_F^2\\
		&\le {}3nc^2\alpha^2\sum_{t=0}^{T-1}\mathbb E\bigl[\lVert\bar m^{(t)}\rVert^2\bigr] +270\alpha^2h^2L^2\sum_{t=0}^{T-1}\mathbb E\bigl[\lVert\mathbf x^{(t+1)}-\mathbf x^{(t)}\rVert_F^2\bigr]\\
		&\quad +\frac{324\alpha^2h^2L^2}{K}\sum_{t=1}^{T-1}\sum_{k=0}^{K-1}\mathbb E\bigl[\lVert\mathbf x^{(t,k)}-\mathbf x^{(t)}\rVert_F^2\bigr] +\frac{216nT\alpha^2h^2(1-\beta)^2\sigma^2}{K}\\
		&\quad +\frac{90n\alpha^2h^2\sigma^2}{b_0} +72n\alpha^2h^2\Delta_fL,
	\end{aligned}
\end{equation}
where the last inequality follows from Lemmas~\ref{lemma12} and~\ref{lemma16} together with the bound $\lVert\nabla f(\mathbf x^{(0)})\rVert_F^2\le2n\Delta_fL$. The result follows by rearranging \eqref{C.3.2} and using $\alpha\le\frac{1}{25hL}$, which implies $1-270\alpha^2h^2L^2\ge\frac12$. \qed
\end{proof}

\subsection{Proof of Lemma~\ref{lemma10}} \label{appendixC.4}

\begin{proof}
For any $i\in\mathcal N$, $t\ge1$, and $k=0,\ldots,K-1$, the local updates give
\begin{equation} \label{C.4.1}
	\begin{aligned}
		&x_i^{(t,k)} - x_i^{(t)}\\
		&= -\eta k y_i^{(t)} + \eta k (1-\beta) m_i^{(t-1)} - \eta \sum_{r=0}^{k-1} g_i(x_i^{(t,r)},\xi_i^{(t,r)}) + \eta \beta \sum_{r=0}^{k-1} g_i(x_i^{(t-1)},\xi_i^{(t,r)})\\
		&= -\eta k y_i^{(t)} + \eta k (1-\beta) \bigl(m_i^{(t-1)} - \nabla f_i(x_i^{(t-1)})\bigr) - \eta(1-\beta)\sum_{r=0}^{k-1} \delta_{1,i}^{(t,r)}\\
		&\quad - \eta \beta \sum_{r=0}^{k-1} (\delta_{1,i}^{(t,r)} - \delta_{2,i}^{(t,r)}) - \eta \sum_{r=0}^{k-1} \Delta_i^{(t,r)} -\eta k \bigl(\nabla f_i (x_i^{(t)}) - \nabla f_i (x_i^{(t-1)})\bigr),
	\end{aligned}
\end{equation}
where $\delta_{1,i}^{(t,r)} = g_i(x_i^{(t)}, \xi_i^{(t,r)}) - \nabla f_i(x_i^{(t)})$, $\delta_{2,i}^{(t,r)} = g_i(x_i^{(t-1)}, \xi_i^{(t,r)}) - \nabla f_i(x_i^{(t-1)})$, and $\Delta_i^{(t,r)} = g_i(x_i^{(t,r)}, \xi_i^{(t,r)}) - g_i(x_i^{(t)}, \xi_i^{(t,r)})$. From \eqref{C.4.1}, it follows that
\begin{equation} \label{C.4.2}
	\begin{aligned}
		&\mathbb E \bigl[\lVert x_i^{(t,k)} - x_i^{(t)} \rVert^2 \bigm| \mathcal F_{t,0}\bigr] \le \: 6\eta^2 k^2 \lVert y_i^{(t)} \rVert^2 \!+ 6\eta^2 k^2 (1-\beta)^2 \lVert m_i^{(t-1)} \!- \nabla f_i(x_i^{(t-1)}) \rVert^2\\
		&\quad + 6\eta^2 (1-\beta)^2 \mathbb E \Bigl[\Bigl\lVert\sum_{r=0}^{k-1} \delta_{1,i}^{(t,r)} \Bigr\rVert^2 \Bigm| \mathcal F_{t,0}\Bigr] + 6\eta^2 \beta^2 \mathbb E \Bigl[\Bigl\lVert\sum_{r=0}^{k-1} (\delta_{1,i}^{(t,r)} - \delta_{2,i}^{(t,r)}) \Bigr\rVert^2 \Bigm| \mathcal F_{t,0}\Bigr]\\
		&\quad + 6\eta^2 \mathbb E \Bigl[\Bigl\lVert\sum_{r=0}^{k-1} \Delta_i^{(t,r)} \Bigr\rVert^2 \Bigm| \mathcal F_{t,0}\Bigr] + 6\eta^2 k^2 \lVert\nabla f_i (x_i^{(t)}) - \nabla f_i (x_i^{(t-1)}) \rVert^2.
	\end{aligned}
\end{equation}

Using the unbiasedness of the stochastic gradients and the independence of the samples, together with the variance bound and mean-squared smoothness, we obtain
\begin{equation} \label{C.4.3}
	\mathbb E\Bigl[\Bigl\lVert\sum_{r=0}^{k-1}\delta_{1,i}^{(t,r)}
	\Bigr\rVert^2\Bigm|\mathcal F_{t,0}\Bigr]
	\le k\sigma^2,
\end{equation}
and
\begin{equation} \label{C.4.4}
	\mathbb E\Bigl[\Bigl\lVert\sum_{r=0}^{k-1}
	(\delta_{1,i}^{(t,r)}-\delta_{2,i}^{(t,r)})
	\Bigr\rVert^2\Bigm|\mathcal F_{t,0}\Bigr]
	\le kL^2\lVert x_i^{(t)}-x_i^{(t-1)}\rVert^2.
\end{equation}

For the fifth term in \eqref{C.4.2}, it follows from \eqref{C.2.8} that
\begin{equation} \label{C.4.5}
	\mathbb E\Bigl[\Bigl\lVert\sum_{r=0}^{k-1}\Delta_i^{(t,r)}
	\Bigr\rVert^2\Bigm|\mathcal F_{t,0}\Bigr]
	\le kL^2\sum_{r=0}^{k-1}
	\mathbb E\bigl[\lVert x_i^{(t,r)}-x_i^{(t)}\rVert^2\bigm|\mathcal F_{t,0}\bigr].
\end{equation}
Moreover, from the $L$-smoothness of $f_i$, it follows that
\begin{equation} \label{C.4.6}
	\lVert\nabla f_i (x_i^{(t)}) - \nabla f_i (x_i^{(t-1)}) \rVert^2 \le L^2
	\lVert x_i^{(t)} - x_i^{(t-1)} \rVert^2.
\end{equation}
Substituting \eqref{C.4.3}--\eqref{C.4.6} into \eqref{C.4.2}, taking expectations, and using $\beta^2\le1$, we obtain
\begin{equation} \label{C.4.7}
	\begin{aligned}
		&\mathbb E \bigl[\lVert x_i^{(t,k)} - x_i^{(t)} \rVert^2 \bigr] \le \: 6\eta^2 k^2 (1-\beta)^2 \mathbb E \bigl[\lVert m_i^{(t-1)} \!- \nabla f_i(x_i^{(t-1)}) \rVert^2\bigr]\\
		&\quad + 6k\eta^2 L^2\sum_{r=0}^{k-1} \mathbb E\bigl[\lVert x_i^{(t,r)}-x_i^{(t)}\rVert^2\bigr] + 12\eta^2 k^2 L^2 \mathbb E\bigl[\lVert x_i^{(t)} - x_i^{(t-1)} \rVert^2\bigr]\\
		&\quad +6\eta^2 k^2 \mathbb E \bigl[\lVert y_i^{(t)} \rVert^2 \bigr]+ 6 k \eta^2 (1-\beta)^2 \sigma^2 .
	\end{aligned}
\end{equation}
Averaging \eqref{C.4.7} over $k=0,\ldots,K-1$ and simplifying the resulting finite sums yields
\begin{equation} \label{C.4.8}
	\begin{aligned}
		&\frac{1}{K}\sum_{k=0}^{K-1}\mathbb E \bigl[\lVert x_i^{(t,k)} - x_i^{(t)} \rVert^2 \bigr] \le \: 2\eta^2 (1-\beta)^2 K^2 \mathbb E \bigl[\lVert m_i^{(t-1)} \!- \nabla f_i(x_i^{(t-1)}) \rVert^2\bigr]\\
		&\quad + 3\eta^2 L^2 K\sum_{k=0}^{K-1} \mathbb E\bigl[\lVert x_i^{(t,k)}-x_i^{(t)}\rVert^2\bigr] + 4\eta^2 L^2 K^2 \mathbb E\bigl[\lVert x_i^{(t)} - x_i^{(t-1)} \rVert^2\bigr]\\
		&\quad +2\eta^2 K^2 \mathbb E \bigl[\lVert y_i^{(t)} \rVert^2 \bigr]+ 3 \eta^2 (1-\beta)^2 K \sigma^2.
	\end{aligned}
\end{equation}
Summing over $i=1,\ldots,n$ and $t=1,\ldots,T-1$, and using $\mathbf y^{(t)}=\mathbf y^{(t)}-\bar{\mathbf y}^{(t)}+n\pi_B\bar m^{(t-1)}$, we obtain
\begin{equation} \label{C.4.9}
	\begin{aligned}
		&\frac{1}{K} \sum_{t=1}^{T-1} \sum_{k=0}^{K-1}\mathbb E \bigl[\lVert {\mathbf x}^{(t,k)} - {\mathbf x}^{(t)} \rVert_F^2 \bigr] \le \: 2\eta^2 (1-\beta)^2 K^2 \sum_{t=0}^{T-1}\mathbb E \bigl[\lVert {\mathbf m}^{(t)} \!- \nabla f({\mathbf x}^{(t)}) \rVert_F^2\bigr]\\
		&\quad + 3\eta^2 L^2 K\sum_{t=1}^{T-1}\sum_{k=0}^{K-1} \mathbb E\bigl[\lVert {\mathbf x}^{(t,k)}-{\mathbf x}^{(t)}\rVert_F^2\bigr] + 4\eta^2 L^2 K^2 \sum_{t=0}^{T-1}\mathbb E\bigl[\lVert {\mathbf x}^{(t+1)} - {\mathbf x}^{(t)} \rVert_F^2\bigr]\\
		&\quad +4\eta^2 K^2 \sum_{t=0}^{T-1} \mathbb E \bigl[\lVert {\mathbf y}^{(t+1)} - {\bar {\mathbf y}}^{(t+1)} \rVert_F^2 \bigr] + 4 n^2 \eta^2 K^2 \sum_{t=0}^{T-1} \mathbb E \bigl[\lVert\bar m ^{(t)} \rVert^2 \bigr] + 3 nT \eta^2 (1-\beta)^2 K \sigma^2.
	\end{aligned}
\end{equation}

Applying Lemmas~\ref{lemma7} and~\ref{lemma15} to \eqref{C.4.9} and using $s_B\ge1$, we obtain
\begin{equation} \label{C.4.10}
	\begin{aligned}
		&\frac{1}{K}\sum_{t=1}^{T-1}\sum_{k=0}^{K-1}\mathbb E\bigl[\lVert\mathbf x^{(t,k)}-\mathbf x^{(t)}\rVert_F^2\bigr] \le{}83s_B^2\eta^2KL^2\sum_{t=1}^{T-1}\sum_{k=0}^{K-1}\mathbb E\bigl[\lVert\mathbf x^{(t,k)}-\mathbf x^{(t)}\rVert_F^2\bigr]\\
		&\quad +70s_B^2\eta^2K^2L^2\sum_{t=0}^{T-1}\mathbb E\bigl[\lVert\mathbf x^{(t+1)}-\mathbf x^{(t)}\rVert_F^2\bigr]+4n^2\eta^2K^2\sum_{t=0}^{T-1}\mathbb E\bigl[\lVert\bar m^{(t)}\rVert^2\bigr]\\
		&\quad+57nT\eta^2Ks_B^2(1-\beta)^2\sigma^2+\frac{22n\eta^2K^2s_B^2\sigma^2}{b_0}+16n\eta^2K^2s_B^2\Delta_fL.
	\end{aligned}
\end{equation}
Substituting the bound in Lemma~\ref{lemma9} into \eqref{C.4.10} gives
\begin{equation} \label{C.4.11}
	\begin{aligned}
		&\frac{1}{K}\sum_{t=1}^{T-1}\sum_{k=0}^{K-1}\mathbb E\bigl[\lVert\mathbf x^{(t,k)}-\mathbf x^{(t)}\rVert_F^2\bigr] \le{}160\eta^2KL^2s_B^2\sum_{t=1}^{T-1}\sum_{k=0}^{K-1}\mathbb E\bigl[\lVert\mathbf x^{(t,k)}-\mathbf x^{(t)}\rVert_F^2\bigr]\\
		&\quad +5n^2\eta^2K^2\sum_{t=0}^{T-1}\mathbb E\bigl[\lVert\bar m^{(t)}\rVert^2\bigr]+33n\eta^2K^2s_B^2\Delta_fL\\
		&\quad +108nT\eta^2Ks_B^2(1-\beta)^2\sigma^2+\frac{43n\eta^2K^2s_B^2\sigma^2}{b_0},
	\end{aligned}
\end{equation}
where we used $c\le n$, $h^2\ge ns_B^2$, and $\alpha^2h^2L^2\le\frac{1}{625}$. Finally, the condition on $\eta$ gives $160\eta^2K^2L^2s_B^2\le\frac{1}{2}$. Rearranging \eqref{C.4.11} and using this bound completes the proof. \qed
\end{proof}

\subsection{Proof of Theorem~\ref{theorem2}} \label{appendixC.5}

\begin{proof}
The bound in \eqref{B.6.1} also holds for LASPP. Applying Lemmas~\ref{lemma7} and~\ref{lemma8} to this bound yields
\begin{equation} \label{C.5.2}
	\begin{aligned}
		&\sum_{t=0}^{T-1}\mathbb E\bigl[\lVert\nabla f(\hat x^{(t)})\rVert^2\bigr] \le{}\frac{4\Delta_f}{c\alpha}-\frac{3}{4}\sum_{t=0}^{T-1}\mathbb E\bigl[\lVert\bar m^{(t)}\rVert^2\bigr]\\
		&+\Bigl(\frac{90s_B^2L^2}{c^2}+\frac{12L^2}{n^2K(1-\beta)}\Bigr)\sum_{t=0}^{T-1}\mathbb E\bigl[\lVert\mathbf x^{(t+1)}-\mathbf x^{(t)}\rVert_F^2\bigr]\\
		&+\Bigl(\frac{108s_B^2L^2}{c^2}+\frac{16L^2}{n(1-\beta)^2}\Bigr)\frac{1}{K}\sum_{t=1}^{T-1}\sum_{k=0}^{K-1}\mathbb E\bigl[\lVert\mathbf x^{(t,k)}-\mathbf x^{(t)}\rVert_F^2\bigr]\\
		&+\frac{4\sigma^2}{nb_0(1-\beta)}+\frac{12T(1-\beta)\sigma^2}{nK}+\frac{72nTs_B^2(1-\beta)^2\sigma^2}{c^2K} +\frac{30ns_B^2\sigma^2}{c^2b_0}+\frac{24ns_B^2\Delta_fL}{c^2}.
	\end{aligned}
\end{equation}

Applying Lemma~\ref{lemma9} to \eqref{C.5.2} gives
\begin{equation} \label{C.5.3}
	\begin{aligned}
		&\sum_{t=0}^{T-1}\mathbb E\bigl[\lVert\nabla f(\hat x^{(t)})\rVert^2\bigr] \le{}\frac{4\Delta_f}{c\alpha} -\frac{1}{4}\sum_{t=0}^{T-1}\mathbb E\bigl[\lVert\bar m^{(t)}\rVert^2\bigr]\\
		&\quad +\Bigl(\frac{132s_B^2L^2}{c^2}+\frac{20L^2}{n(1-\beta)^2}\Bigr) \frac{1}{K}\sum_{t=1}^{T-1}\sum_{k=0}^{K-1} \mathbb E\bigl[\lVert\mathbf x^{(t,k)}-\mathbf x^{(t)}\rVert_F^2\bigr]\\
		&\quad +\frac{4\sigma^2}{nb_0(1-\beta)} +\frac{12T(1-\beta)\sigma^2}{nK} +\frac{101Th^2(1-\beta)^2\sigma^2}{c^2K} +\frac{42h^2\sigma^2}{c^2b_0} +\frac{34h^2\Delta_fL}{c^2},
	\end{aligned}
\end{equation}
where we used $ns_B^2\le h^2$, $\frac{1-\beta}{nK}\le1$, and the step-size bounds $\alpha^2h^2L^2\le\frac{1}{2500}$ and $\frac{c^2\alpha^2L^2}{nK(1-\beta)}\le\frac{1}{400}$. In particular, the coefficient of $\sum_{t=0}^{T-1}\mathbb E\bigl[\lVert\bar m^{(t)}\rVert^2\bigr]$ satisfies
$
540ns_B^2\alpha^2L^2
+\frac{72c^2\alpha^2L^2}{nK(1-\beta)}
-\frac{3}{4}
\le-\frac{1}{4}.
$

Applying Lemma~\ref{lemma10} to \eqref{C.5.3} yields
\begin{equation} \label{C.5.4}
	\begin{aligned}
		&\sum_{t=0}^{T-1}\mathbb E\bigl[\lVert\nabla f(\hat x^{(t)})\rVert^2\bigr] \le{}\frac{4\Delta_f}{c\alpha} +\frac{4\sigma^2}{nb_0(1-\beta)} +\frac{12T(1-\beta)\sigma^2}{nK}\\
		&\quad +\frac{105Th^2(1-\beta)^2\sigma^2}{c^2K} +\frac{45h^2\sigma^2}{c^2b_0}+\frac{36h^2\Delta_fL}{c^2},
	\end{aligned}
\end{equation}
where we used the bounds on $\eta$, together with $ns_B^2\le h^2$ and $\frac{s_B^2}{n}\le\frac{h^2}{c^2}$. In particular, the coefficient of $\sum_{t=0}^{T-1}\mathbb E\bigl[\lVert\bar m^{(t)}\rVert^2\bigr]$ satisfies
$
\frac{1320n^2\eta^2K^2s_B^2L^2}{c^2}
+\frac{200n\eta^2K^2L^2}{(1-\beta)^2}
-\frac{1}{4}\le0.
$
Dividing \eqref{C.5.4} by $T$ completes the proof. \qed
\end{proof}

\section{Proof of Supporting Lemmas}

\subsection{Proof of Lemma~\ref{lemma12}} \label{appendixD.1}

\begin{proof}
By the unbiasedness of the stochastic gradients and the independence of the initialization samples, we have
\begin{equation} \label{D.1.1}
	\begin{aligned}
		&\mathbb E\bigl[\lVert\bar m^{(0)} - \overline {\nabla f}(\mathbf x^{(0)})\rVert^2\bigr]\\
		&= \mathbb E \Bigl[ \Bigl\lVert\frac{1}{n} \sum_{i=1}^n \frac{1}{b_0} \sum_{r=1}^{b_0}\bigl( g_i(x_{i}^{(0)}, \xi_{i,r}^{(0)}) - \nabla f_i (x_i^{(0)}) \bigr) \Bigr\rVert^2 \Bigr]\\
		&= \frac{1}{n^2 b_0^2} \sum_{i=1}^n \sum_{r=1}^{b_0} \mathbb E\bigl[\bigl\lVert\bigl( g_i(x_{i}^{(0)}, \xi_{i,r}^{(0)}) - \nabla f_i (x_i^{(0)}) \bigr) \bigr\rVert^2 \bigr]\le \frac{\sigma^2}{n b_0},
	\end{aligned}
\end{equation}
and
\begin{equation} \label{D.1.2}
	\begin{aligned}
		&\mathbb E \bigl[\lVert\mathbf{m}^{(0)} - {\nabla f}(\mathbf x^{(0)}) \rVert_F^2 \bigr]\\
		&= \: \sum_{i=1}^n\mathbb E \Bigl[ \Bigl\lVert\frac{1}{b_0} \sum_{r=1}^{b_0} \bigl( g_i(x_i^{(0)}, \xi_{i,r}^{(0)}) -\nabla f_i (x_i^{(0)}) \bigr) \Bigr\rVert^2 \Bigr]\\
		&= \: \frac{1}{b_0^2} \sum_{i=1}^n\sum_{r=1}^{b_0} \mathbb E\bigl[\bigl\lVert\bigl( g_i(x_{i}^{(0)}, \xi_{i,r}^{(0)}) - \nabla f_i (x_i^{(0)}) \bigr) \bigr\rVert^2 \bigr]\le \frac{n\sigma^2}{b_0},
	\end{aligned}
\end{equation}
where the inequalities follow from the variance bound. This completes the proof. \qed
\end{proof}

\subsection{Proof of Lemma~\ref{lemma13}} \label{appendixD.2}

\begin{proof}
Recalling \eqref{B.4.1}, we have, for any $i\in\mathcal N$ and $t\ge1$,
\begin{equation} \label{D.2.1}
	\begin{aligned}
		&m_i^{(t)} - \nabla f_i(x_i^{(t)}) = \beta (m_i^{(t-1)} - \nabla f_i(x_i^{(t-1)})) + (1-\beta) \delta_{1,i}^{(t)} + \beta (\delta_{1,i}^{(t)} - \delta_{2,i}^{(t)}),
	\end{aligned}
\end{equation}
where $\delta_{1,i}^{(t)} = g_i(x_i^{(t)}, \xi_i^{(t)}) - \nabla f_i(x_i^{(t)})$ and $\delta_{2,i}^{(t)} = g_i(x_i^{(t-1)}, \xi_i^{(t)}) - \nabla f_i(x_i^{(t-1)})$. From \eqref{D.2.1}, it follows that
\begin{equation} \label{D.2.2}
	\begin{aligned}
		&\mathbb E\bigl[\lVert m_i^{(t)} - \nabla f_i(x_i^{(t)}) \rVert^2 \bigm| \mathcal F_t\bigr]\\
		&= \:\beta^2 \lVert m_i^{(t-1)} - \nabla f_i(x_i^{(t-1)})\rVert^2 + \mathbb E \bigl[ \lVert (1-\beta) \delta_{1,i}^{(t)} + \beta (\delta_{1,i}^{(t)} - \delta_{2,i}^{(t)}) \rVert^2 \bigm| \mathcal F_t \bigr]\\
		&\le \:\beta^2 \lVert m_i^{(t-1)} - \nabla f_i(x_i^{(t-1)})\rVert^2 + {2(1-\beta)^2}\mathbb E \bigl[ \lVert\delta_{1,i}^{(t)} \rVert^2 \bigm| \mathcal F_t \bigr] + {2 \beta^2}\mathbb E \bigl[ \lVert (\delta_{1,i}^{(t)} - \delta_{2,i}^{(t)}) \rVert^2 \bigm| \mathcal F_t \bigr],
	\end{aligned}
\end{equation}
where the equality follows from the unbiasedness of the stochastic gradients and the $\mathcal F_t$-measurability of $m_i^{(t-1)}-\nabla f_i(x_i^{(t-1)})$.

For the second term in \eqref{D.2.2}, the variance bound gives
\begin{equation} \label{D.2.3}
	\mathbb E \bigl[ \lVert\delta_{1,i}^{(t)} \rVert^2  \bigm| \mathcal F_t \bigr] = \mathbb E \bigl[ \lVert g_i(x_i^{(t)}, \xi_i^{(t)}) - \nabla f_i(x_i^{(t)}) \rVert^2  \bigm| \mathcal F_t \bigr] \le \sigma^2.
\end{equation}
For the third term, the unbiasedness and mean-squared smoothness of $g_i$ give
\begin{equation} \label{D.2.4}
	\begin{aligned}
		&\mathbb E \bigl[ \lVert (\delta_{1,i}^{(t)} - \delta_{2,i}^{(t)}) \rVert^2 \bigm| \mathcal F_t \bigr]\\
		&= \: \mathbb E \bigl[ \lVert g_i(x_i^{(t)}, \xi_i^{(t)}) - g_i(x_i^{(t-1)}, \xi_i^{(t)}) \rVert^2 \bigm|\mathcal F_t\bigr] - \lVert\nabla f_i(x_i^{(t)}) - \nabla f_i(x_i^{(t-1)}) \rVert^2\\
		&\le \: L^2 \lVert x_i^{(t)} - x_i^{(t-1)} \rVert^2.
	\end{aligned}
\end{equation}

Substituting \eqref{D.2.3} and \eqref{D.2.4} into \eqref{D.2.2} and taking the total expectation yields
\begin{equation} \label{D.2.5}
	\begin{aligned}
		&\mathbb E\bigl[\lVert m_i^{(t)} - \nabla f_i(x_i^{(t)}) \rVert^2 \bigr] \le \:\beta^2 \mathbb E\bigl[\lVert m_i^{(t-1)} - \nabla f_i(x_i^{(t-1)})\rVert^2\bigr] + {2(1-\beta)^2} \sigma^2\\
		&\quad + {2 \beta^2}L^2 \mathbb E\bigl[\lVert x_i^{(t)} - x_i^{(t-1)} \rVert^2\bigr].
	\end{aligned}
\end{equation}
Unrolling \eqref{D.2.5} and summing over $i=1,\ldots,n$ and $t=0,\ldots,T-1$, we obtain
\begin{equation} \label{D.2.6}
	\begin{aligned}
		&\sum_{t=0}^{T-1} \mathbb E \bigl[\lVert {\mathbf m}^{(t)} - {\nabla f} (\mathbf x^{(t)}) \rVert_F^2\bigr] \le \: \frac{1}{1 - \beta^2}\mathbb E\bigl[\lVert {\mathbf m}^{(0)} - {\nabla f} (\mathbf x^{(0)}) \rVert^2_F\bigr] + \frac{2nT(1-\beta)^2 \sigma^2}{1-\beta^2}\\
		&\quad + \frac{2 \beta^2 L^2}{1 - \beta^2} \sum_{t=0}^{T-1} \mathbb E \bigl[\lVert {\mathbf x}^{(t+1)} - {\mathbf x}^{(t)} \rVert_F^2\bigr],
	\end{aligned}	
\end{equation}
where we interchanged the order of summation and bounded the resulting geometric sums. Applying Lemma~\ref{lemma12} to \eqref{D.2.6} and using $\beta\in[0,1)$ completes the proof. \qed
\end{proof}

\subsection{Proof of Lemma~\ref{lemma14}} \label{appendixD.3}

\begin{proof}
From the update of $m_i^{(t)}$ in \eqref{eq3}, we have, for $t\ge1$ and $i\in\mathcal N$,
\begin{equation} \label{D.3.1}
	\begin{aligned}
		&m_i^{(t)} - m_i^{(t-1)} = \: g_i(x_i^{(t)}, \xi_i^{(t)}) - g_i(x_i^{(t-1)}, \xi_i^{(t)}) - (1-\beta) \bigl(m_i^{(t-1)} - \nabla f_i(x_i^{(t-1)})\bigr)\\
		&\quad + (1-\beta) \bigl(g_i(x_i^{(t-1)}, \xi_i^{(t)}) - \nabla f_i(x_i^{(t-1)})\bigr).
	\end{aligned}
\end{equation}
From \eqref{D.3.1}, it follows that
\begin{equation} \label{D.3.2}
	\begin{aligned}
		&\mathbb E\bigl[\lVert m_i^{(t)} - m_i^{(t-1)} \rVert^2 \bigm| \mathcal F_t\bigr]\\
		&\le \: 3\mathbb E\bigl[\lVert g_i(x_i^{(t)}, \xi_i^{(t)}) - g_i(x_i^{(t-1)}, \xi_i^{(t)}) \rVert^2 \bigm| \mathcal F_t\bigr] + 3(1-\beta)^2 \lVert m_i^{(t-1)} - \nabla f_i(x_i^{(t-1)}) \rVert^2\\
		&\quad + 3(1-\beta)^2 \mathbb E\bigl[\lVert g_i(x_i^{(t-1)}, \xi_i^{(t)}) - \nabla f_i(x_i^{(t-1)}) \rVert^2 \bigm| \mathcal F_t\bigr]\\
		&\le \: 3L^2 \lVert x_i^{(t)} - x_i^{(t-1)} \rVert^2 + 3(1-\beta)^2 \sigma^2 + 3(1-\beta)^2 \lVert m_i^{(t-1)} - \nabla f_i(x_i^{(t-1)}) \rVert^2,
	\end{aligned}
\end{equation}
where the last inequality follows from the mean-squared smoothness and bounded-variance assumptions on the stochastic gradients.

Taking total expectations and summing over $i=1,\ldots,n$ and $t=1,\ldots,T-1$ yields
\begin{equation} \label{D.3.3}
	\begin{aligned}
		&\sum_{t=1}^{T-1}\mathbb E\bigl[\lVert\mathbf m^{(t)} - \mathbf m^{(t-1)} \rVert_F^2\bigr] \le \: 3L^2 \sum_{t=0}^{T-1}\mathbb E\bigl[\lVert\mathbf x^{(t+1)} - \mathbf x^{(t)} \rVert_F^2\bigr] + 3nT(1-\beta)^2 \sigma^2\\
		&\quad + 3(1-\beta)^2 \sum_{t=0}^{T-1}\mathbb E\bigl[\lVert\mathbf m^{(t)} - \nabla f(\mathbf x^{(t)}) \rVert_F^2\bigr].
	\end{aligned}
\end{equation}
Substituting the bound from Lemma~\ref{lemma13} into \eqref{D.3.3} and using $\beta\in[0,1)$ completes the proof. \qed
\end{proof}

\subsection{Proof of Lemma~\ref{lemma15}} \label{appendixD.4}

\begin{proof}
For $t\ge1$ and $i\in\mathcal N$, it follows from \eqref{C.2.1} that
\begin{equation} \label{D.4.1}
	\begin{aligned}
		&m_i^{(t)} - \nabla f_i(x_i^{(t)}) = \: \beta \bigl(m_i^{(t-1)} - \nabla f_i(x_i^{(t-1)})\bigr) + \frac{1-\beta}{K} \sum_{k=0}^{K-1} \delta_{1,i}^{(t,k)}\\
		&\quad + \frac{\beta}{K} \sum_{k=0}^{K-1}(\delta_{1,i}^{(t,k)} - \delta_{2,i}^{(t,k)}) + \frac{1}{K}\sum_{k=0}^{K-1}\Delta_i^{(t,k)},
	\end{aligned}
\end{equation}
where $\delta_{1,i}^{(t,k)} = g_i(x_i^{(t)}, \xi_i^{(t,k)}) - \nabla f_i(x_i^{(t)})$, $\delta_{2,i}^{(t,k)} = g_i(x_i^{(t-1)}, \xi_i^{(t,k)}) - \nabla f_i(x_i^{(t-1)})$, and $\Delta_i^{(t,k)} = g_i(x_i^{(t,k)}, \xi_i^{(t,k)}) - g_i(x_i^{(t)}, \xi_i^{(t,k)})$. 

Using the unbiasedness of the stochastic gradients and the $\mathcal F_{t,0}$-measurability of $m_i^{(t-1)}-\nabla f_i(x_i^{(t-1)})$, we obtain from \eqref{D.4.1}
\begin{equation} \label{D.4.2}
	\begin{aligned}
		&\mathbb E\bigl[\lVert m_i^{(t)} - {\nabla f}_i(x_i^{(t)})\rVert^2 \bigm| \mathcal F_{t,0}\bigr] \le \: \beta^2 \lVert m_i^{(t-1)} - \nabla f_i(x_i^{(t-1)}) \rVert^2\\
		&\quad + \frac{3(1-\beta)^2}{K^2} \mathbb E \Bigl[ \Bigl\lVert\sum_{k=0}^{K-1} \delta_{1,i}^{(t,k)} \Bigr\rVert^2 \Bigm| \mathcal F_{t,0} \Bigr] + \frac{3 \beta^2}{K^2} \mathbb E \Bigl[ \Bigl\lVert\sum_{k=0}^{K-1}(\delta_{1,i}^{(t,k)} - \delta_{2,i}^{(t,k)}) \Bigr\rVert^2 \Bigm| \mathcal F_{t,0} \Bigr]\\
		&\quad + \frac{3}{K^2}\mathbb E \Bigl[ \Bigl\lVert\sum_{k=0}^{K-1} \Delta_i^{(t,k)} \Bigr\rVert^2 \Bigm| \mathcal F_{t,0} \Bigr] + 2\Bigl\langle \beta \bigl(m_i^{(t-1)} - \nabla f_i(x_i^{(t-1)})\bigr), \frac{1}{K} \mathbb E \Bigl[\sum_{k=0}^{K-1}\! \Delta_i^{(t,k)} \Bigm| \mathcal F_{t,0} \Bigr] \Bigr\rangle.
	\end{aligned}
\end{equation}

For the second and third terms on the right-hand side of \eqref{D.4.2}, the unbiasedness of the stochastic gradients and the independence of the samples give
\begin{equation} \label{D.4.3}
	\begin{aligned}
		&\mathbb E \Bigl[ \Bigl\lVert\sum_{k=0}^{K-1} \delta_{1,i}^{(t,k)} \Bigr\rVert^2 \Bigm| \mathcal F_{t,0} \Bigr] = \sum_{k=0}^{K-1} \mathbb E \bigl[\lVert\delta_{1,i}^{(t,k)} \rVert^2 \bigm| \mathcal F_{t,0}\bigr] \le K \sigma^2,
	\end{aligned}
\end{equation}
and
\begin{equation} \label{D.4.5}
	\begin{aligned}
		&\mathbb E\Bigl[ \Bigl\lVert \sum_{k=0}^{K-1} \bigl(\delta_{1,i}^{(t,k)}-\delta_{2,i}^{(t,k)}\bigr) \Bigr\rVert^2 \Bigm|\mathcal F_{t,0} \Bigr]\\
		&= {}\sum_{k=0}^{K-1} \mathbb E\bigl[ \lVert\delta_{1,i}^{(t,k)}-\delta_{2,i}^{(t,k)}\rVert^2 \bigm|\mathcal F_{t,0} \bigr]\\
		&= {}\sum_{k=0}^{K-1} \mathbb E\bigl[ \lVert g_i(x_i^{(t)},\xi_i^{(t,k)}) -g_i(x_i^{(t-1)},\xi_i^{(t,k)}) \rVert^2 \bigm|\mathcal F_{t,0} \bigr]-K\lVert \nabla f_i(x_i^{(t)})-\nabla f_i(x_i^{(t-1)}) \rVert^2\\
		&\le {}KL^2\lVert x_i^{(t)}-x_i^{(t-1)}\rVert^2.
	\end{aligned}
\end{equation}
The final inequalities in \eqref{D.4.3} and \eqref{D.4.5} follow from the bounded variance and mean-squared smoothness of the stochastic gradients, respectively.

For the fourth term in \eqref{D.4.2}, we have
\begin{equation} \label{D.4.6}
	\begin{aligned}
		&\mathbb E \Bigl[ \Bigl\lVert\sum_{k=0}^{K-1} \Delta_i^{(t,k)} \Bigr\rVert^2 \Bigm| \mathcal F_{t,0} \Bigr]\\
		&\le \: K \sum_{k=0}^{K-1} \mathbb E \bigl[ \bigl\lVert g_i(x_i^{(t,k)}\!, \xi_i^{(t,k)}) \! - g_i(x_i^{(t)}\!, \xi_i^{(t,k)}) \bigr\rVert^2 \bigm| \mathcal F_{t,0} \bigr]\\
		&\le \: K L^2 \sum_{k=0}^{K-1} \mathbb E \bigl[ \lVert x_i^{(t,k)} - x_i^{(t)} \rVert^2 \bigm|\mathcal F_{t,0}\bigr],
	\end{aligned}
\end{equation}
where the last inequality follows from \eqref{C.2.8}. 

For the last term in \eqref{D.4.2}, Young's inequality and Jensen's inequality give
\begin{equation} \label{D.4.7}
	\begin{aligned}
		&2\Bigl\langle \beta \bigl(m_i^{(t-1)} - \nabla f_i(x_i^{(t-1)})\bigr), \frac{1}{K} \mathbb E \Bigl[\sum_{k=0}^{K-1}\! \Delta_i^{(t,k)} \Bigm| \mathcal F_{t,0} \Bigr] \Bigr\rangle\\
		&\le \: (1-\beta) \beta^2 \lVert m_i^{(t-1)} - \nabla f_i(x_i^{(t-1)}) \rVert^2 + \frac{1}{K^2 (1-\beta)} \mathbb E \Bigl[ \Bigl\lVert\sum_{k=0}^{K-1} \Delta_i^{(t,k)} \Bigr\rVert^2 \Bigm| \mathcal F_{t,0} \Bigr]\\
		&\le \: (1-\beta) \beta^2 \lVert m_i^{(t-1)} - \nabla f_i(x_i^{(t-1)}) \rVert^2 + \frac{L^2}{K (1-\beta)} \sum_{k=0}^{K-1} \mathbb E \bigl[ \lVert x_i^{(t,k)} - x_i^{(t)} \rVert^2 \bigm|\mathcal F_{t,0}\bigr],
	\end{aligned}
\end{equation}
where the last inequality is due to \eqref{D.4.6}.

Substituting \eqref{D.4.3}, \eqref{D.4.5}, \eqref{D.4.6}, and \eqref{D.4.7} into \eqref{D.4.2}, taking the total expectation, and using $\beta\in[0,1)$ yields
\begin{equation} \label{D.4.8}
	\begin{aligned}
		&\mathbb E\bigl[\lVert m_i^{(t)} - {\nabla f}_i(x_i^{(t)})\rVert^2\bigr] \le \: \beta \mathbb E\bigl[\lVert m_i^{(t-1)} - \nabla f_i(x_i^{(t-1)}) \rVert^2\bigr]\\
		&+ \frac{4 L^2}{K (1-\beta)} \sum_{k=0}^{K-1} \mathbb E \bigl[ \lVert x_i^{(t,k)} - x_i^{(t)} \rVert^2 \bigr] + \frac{3 \beta^2 L^2}{K} \mathbb E\bigl[\lVert x_i^{(t)} - x_i^{(t-1)} \rVert^2\bigr] + \frac{3(1-\beta)^2 \sigma^2}{K}.
	\end{aligned}
\end{equation}
Unrolling \eqref{D.4.8} and summing over $i=1,\ldots,n$ and $t=0,\ldots,T-1$, we obtain
\begin{equation} \label{D.4.9}
	\begin{aligned}
		&\sum_{t=0}^{T-1} \mathbb E \bigl[\lVert {\mathbf m}^{(t)} - {\nabla f} (\mathbf x^{(t)}) \rVert_F^2\bigr] \le \: \frac{1}{1 - \beta}\mathbb E\bigl[\lVert {\mathbf m}^{(0)} - {\nabla f} (\mathbf x^{(0)}) \rVert^2_F\bigr] + \frac{3nT(1-\beta) \sigma^2}{K}\\
		&+ \frac{4 L^2}{K (1-\beta)^2} \sum_{t=1}^{T-1} \sum_{k=0}^{K-1} \mathbb E \bigl[ \lVert {\mathbf x}^{(t,k)} - {\mathbf x}^{(t)} \rVert_F^2 \bigr] + \frac{3 \beta^2 L^2}{K(1-\beta)} \sum_{t=0}^{T-1} \mathbb E \bigl[\lVert {\mathbf x}^{(t+1)} - {\mathbf x}^{(t)} \rVert_F^2\bigr],
	\end{aligned}	
\end{equation}
where we interchanged the order of summation and bounded the resulting geometric sums. Applying Lemma~\ref{lemma12} to \eqref{D.4.9} and using $\beta^2\le1$ completes the proof. \qed
\end{proof}

\subsection{Proof of Lemma~\ref{lemma16}} \label{appendixD.5}

\begin{proof}
From \eqref{eq5}, we have, for $t\ge1$ and $i\in\mathcal N$,
\begin{equation} \label{D.5.1}
	\begin{aligned}
		&m_i^{(t)}-m_i^{(t-1)} ={}\frac{1}{K}\sum_{k=0}^{K-1}\bigl(g_i(x_i^{(t,k)},\xi_i^{(t,k)})-g_i(x_i^{(t-1)},\xi_i^{(t,k)})\bigr)\\
		&-(1-\beta)\bigl(m_i^{(t-1)}-\nabla f_i(x_i^{(t-1)})\bigr) +\frac{1-\beta}{K}\sum_{k=0}^{K-1}\bigl(g_i(x_i^{(t-1)},\xi_i^{(t,k)})-\nabla f_i(x_i^{(t-1)})\bigr).
	\end{aligned}
\end{equation}
For $k=0,\ldots,K-1$, let $u_i^{(t,k)}:=g_i(x_i^{(t,k)},\xi_i^{(t,k)})-g_i(x_i^{(t-1)},\xi_i^{(t,k)})$ and $v_i^{(t,k)}:=g_i(x_i^{(t-1)},\xi_i^{(t,k)})-\nabla f_i(x_i^{(t-1)})$. From \eqref{D.5.1}, it follows that
\begin{equation} \label{D.5.2}
	\begin{aligned}
		&\mathbb E\bigl[\lVert m_i^{(t)}-m_i^{(t-1)}\rVert^2\bigm|\mathcal F_{t,0}\bigr] \le{}3(1-\beta)^2\lVert m_i^{(t-1)}-\nabla f_i(x_i^{(t-1)})\rVert^2\\
		&\quad +\frac{3}{K^2}\mathbb E\Bigl[\Bigl\lVert\sum_{k=0}^{K-1}u_i^{(t,k)}\Bigr\rVert^2\Bigm|\mathcal F_{t,0}\Bigr] +\frac{3(1-\beta)^2}{K^2}\mathbb E\Bigl[\Bigl\lVert\sum_{k=0}^{K-1}v_i^{(t,k)}\Bigr\rVert^2\Bigm|\mathcal F_{t,0}\Bigr].
	\end{aligned}
\end{equation}

For the second term in \eqref{D.5.2}, we have
\begin{equation} \label{D.5.3}
	\begin{aligned}
		\mathbb E\Bigl[\Bigl\lVert\sum_{k=0}^{K-1}u_i^{(t,k)}\Bigr\rVert^2\Bigm|\mathcal F_{t,0}\Bigr]
		&\le {}K\sum_{k=0}^{K-1}\mathbb E\bigl[\lVert u_i^{(t,k)}\rVert^2\bigm|\mathcal F_{t,0}\bigr]\\
		&= {}K\sum_{k=0}^{K-1}\mathbb E\Bigl[\mathbb E\bigl[\lVert u_i^{(t,k)}\rVert^2\bigm|\mathcal F_{t,k}\bigr]\Bigm|\mathcal F_{t,0}\Bigr]\\
		&\le {}KL^2\sum_{k=0}^{K-1}\mathbb E\bigl[\lVert x_i^{(t,k)}-x_i^{(t-1)}\rVert^2\bigm|\mathcal F_{t,0}\bigr],
	\end{aligned}
\end{equation}
where the equality follows from the tower property of conditional expectation and the last inequality follows from mean-squared smoothness.

For the third term in \eqref{D.5.2}, we have
\begin{equation} \label{D.5.5}
	\begin{aligned}
		&\mathbb E\Bigl[\Bigl\lVert\sum_{k=0}^{K-1}v_i^{(t,k)}\Bigr\rVert^2\Bigm|\mathcal F_{t,0}\Bigr] ={}\sum_{k=0}^{K-1}\mathbb E\bigl[\lVert v_i^{(t,k)}\rVert^2\bigm|\mathcal F_{t,0}\bigr]\le K\sigma^2,
	\end{aligned}
\end{equation}
where the equality follows from the unbiasedness of the stochastic gradients and the independence of the samples, and the inequality follows from the variance bound.

Substituting \eqref{D.5.3} and \eqref{D.5.5} into \eqref{D.5.2} and taking the total expectation yields
\begin{equation} \label{D.5.7}
	\begin{aligned}
		&\mathbb E\bigl[\lVert m_i^{(t)}-m_i^{(t-1)}\rVert^2\bigr]\\
		&\le {}3(1-\beta)^2\mathbb E\bigl[\lVert m_i^{(t-1)}-\nabla f_i(x_i^{(t-1)})\rVert^2\bigr] +\frac{3L^2}{K}\sum_{k=0}^{K-1}\mathbb E\bigl[\lVert x_i^{(t,k)}-x_i^{(t-1)}\rVert^2\bigr] +\frac{3(1-\beta)^2\sigma^2}{K}\\
		&\le {}3(1-\beta)^2\mathbb E\bigl[\lVert m_i^{(t-1)}-\nabla f_i(x_i^{(t-1)})\rVert^2\bigr] +\frac{6L^2}{K}\sum_{k=0}^{K-1}\mathbb E\bigl[\lVert x_i^{(t,k)}-x_i^{(t)}\rVert^2\bigr]\\
		&\quad +6L^2\mathbb E\bigl[\lVert x_i^{(t)}-x_i^{(t-1)}\rVert^2\bigr] +\frac{3(1-\beta)^2\sigma^2}{K}.
	\end{aligned}
\end{equation}
Summing \eqref{D.5.7} over $i=1,\ldots,n$ and $t=1,\ldots,T-1$, applying Lemma~\ref{lemma15}, and using $K\ge1$ and $\beta\in[0,1)$ completes the proof. \qed
\end{proof}

\bibliographystyle{spmpsci}
\bibliography{ASPP}

@article{arjevani2023lower,
  title={Lower bounds for non-convex stochastic optimization},
  author={Arjevani, Yossi and Carmon, Yair and Duchi, John C and Foster, Dylan J and Srebro, Nathan and Woodworth, Blake},
  journal={Mathematical Programming},
  volume={199},
  number={1--2},
  pages={165--214},
  year={2023},
  publisher={Springer},
}

@book{nesterov2018lectures,
  title={Lectures on {C}onvex {O}ptimization},
  author={Nesterov, Yurii},
  volume={137},
  year={2018},
  publisher={Springer},
  series={Springer Optimization and Its Applications},
  edition={2},
}

@article{liang2025linear,
  title={On the Linear Speedup of the {Push-Pull} Method for Decentralized Optimization over Digraphs},
  author={Liang, Liyuan and Luo, Gan and Yuan, Kun},
  journal={arXiv preprint arXiv:2506.18075},
  year={2025},
  url={https://arxiv.org/abs/2506.18075}
}

@article{xin2018linear,
  title={A linear algorithm for optimization over directed graphs with geometric convergence},
  author={Xin, Ran and Khan, Usman A},
  journal={IEEE Control Systems Letters},
  volume={2},
  number={3},
  pages={315--320},
  year={2018},
  publisher={IEEE},
}

@article{pu2020push,
  title={{Push--Pull} gradient methods for distributed optimization in networks},
  author={Pu, Shi and Shi, Wei and Xu, Jinming and Nedi{\'c}, Angelia},
  journal={IEEE Transactions on Automatic Control},
  volume={66},
  number={1},
  pages={1--16},
  year={2021},
  publisher={IEEE},
}

@article{mai2019distributed,
  title={Distributed optimization over directed graphs with row stochasticity and constraint regularity},
  author={Mai, Van Sy and Abed, Eyad H},
  journal={Automatica},
  volume={102},
  pages={94--104},
  year={2019},
  publisher={Elsevier},
}

@article{nedic2014distributed,
  title={Distributed optimization over time-varying directed graphs},
  author={Nedi{\'c}, Angelia and Olshevsky, Alex},
  journal={IEEE Transactions on Automatic Control},
  volume={60},
  number={3},
  pages={601--615},
  year={2015},
  publisher={IEEE},
}

@article{nedic2017achieving,
  title={Achieving geometric convergence for distributed optimization over time-varying graphs},
  author={Nedi{\'c}, Angelia and Olshevsky, Alex and Shi, Wei},
  journal={SIAM Journal on Optimization},
  volume={27},
  number={4},
  pages={2597--2633},
  year={2017},
  publisher={SIAM},
}

@article{xi2018linear,
  title={Linear convergence in optimization over directed graphs with row-stochastic matrices},
  author={Xi, Chenguang and Mai, Van Sy and Xin, Ran and Abed, Eyad H and Khan, Usman A},
  journal={IEEE Transactions on Automatic Control},
  volume={63},
  number={10},
  pages={3558--3565},
  year={2018},
  publisher={IEEE},
}

@article{you2025stochastic,
  author={You, Runze and Pu, Shi},
  journal={IEEE Transactions on Signal Processing}, 
  title={Stochastic {Push--Pull} for Decentralized Nonconvex Optimization},
  year={2026},
  volume={74},
  pages={1383--1398},
 }

@article{huang2025accelerated,
  title={An accelerated distributed stochastic gradient method with momentum},
  author={Huang, Kun and Pu, Shi and Nedi{\'c}, Angelia},
  journal={Mathematical Programming},
  volume={215},
  number={1--2},
  pages={193--236},
  year={2026},
  publisher={Springer},
}

@inproceedings{cutkosky2019momentum,
  title={Momentum-based variance reduction in non-convex {SGD}},
  author={Cutkosky, Ashok and Orabona, Francesco},
  booktitle={Advances in Neural Information Processing Systems ({NeurIPS})},
  volume={32},
  year={2019},
}

@article{tran2022hybrid,
  title={A hybrid stochastic optimization framework for composite nonconvex optimization},
  author={Tran-Dinh, Quoc and Pham, Nhan H and Phan, Dzung T and Nguyen, Lam M},
  journal={Mathematical Programming},
  volume={191},
  number={2},
  pages={1005--1071},
  year={2022},
  publisher={Springer},
}

@inproceedings{xin2021hybrid,
  title={A hybrid variance-reduced method for decentralized stochastic non-convex optimization},
  author={Xin, Ran and Khan, Usman and Kar, Soummya},
  booktitle={International Conference on Machine Learning ({ICML})},
  pages={11459--11469},
  year={2021},
  volume={139},
}

@InProceedings{liang25c,
  title = 	 {Achieving Linear Speedup and Near-Optimal Complexity for Decentralized Optimization over Row-stochastic Networks},
  author =       {Liang, Liyuan and Chen, Xinyi and Luo, Gan and Yuan, Kun},
  booktitle={International Conference on Machine Learning ({ICML})},
  pages = 	 {37100--37130},
  year = 	 {2025},
  volume={267},
}

@InProceedings{SCAFFOLD,
  title = 	 {{SCAFFOLD}: Stochastic Controlled Averaging for Federated Learning},
  author =       {Karimireddy, Sai Praneeth and Kale, Satyen and Mohri, Mehryar and Reddi, Sashank and Stich, Sebastian and Suresh, Ananda Theertha},
  booktitle={International Conference on Machine Learning ({ICML})},
  pages = 	 {5132--5143},
  year = 	 {2020},
  volume={119},
}

@InProceedings{fedavg,
  title = 	 {{Communication-Efficient Learning of Deep Networks from Decentralized Data}},
  author = 	 {McMahan, Brendan and Moore, Eider and Ramage, Daniel and Hampson, Seth and Arcas, Blaise Aguera y},
  booktitle={International Conference on Artificial Intelligence and Statistics ({AISTATS})},
  pages = 	 {1273--1282},
  year = 	 {2017},
  volume={54},
}

@inproceedings{you2024b,
  title={{B}-ary {Tree Push-Pull} method is provably efficient for distributed learning on heterogeneous data},
  author={You, Runze and Pu, Shi},
  booktitle={Advances in Neural Information Processing Systems ({NeurIPS})},
  volume={37},
  pages={97523--97561},
  year={2024},
}

@inproceedings{song2022communication,
  title={Communication-efficient topologies for decentralized learning with {$O(1)$} consensus rate},
  author={Song, Zhuoqing and Li, Weijian and Jin, Kexin and Shi, Lei and Yan, Ming and Yin, Wotao and Yuan, Kun},
  booktitle={Advances in Neural Information Processing Systems ({NeurIPS})},
  volume={35},
  pages={1073--1085},
  year={2022},
}

@article{Nedich03092025,
author={Nedi{\'c}, Angelia and Nguyen, Duong Thuy Anh and Nguyen, Duong Tung},
title={{AB}/{Push-Pull} method for distributed optimization in time-varying directed networks},
journal = {Optimization Methods and Software},
volume = {40},
number = {5},
pages = {1044--1071},
year = {2025},
}

@inproceedings{lian2017can,
  title     = {Can Decentralized Algorithms Outperform Centralized Algorithms?
               A Case Study for Decentralized Parallel Stochastic Gradient Descent},
  author    = {Lian, Xiangru and Zhang, Ce and Zhang, Huan and
               Hsieh, Cho-Jui and Zhang, Wei and Liu, Ji},
  booktitle={Advances in Neural Information Processing Systems ({NeurIPS})},
  volume    = {30},
  pages     = {5330--5340},
  year      = {2017}
}

@inproceedings{assran2019stochastic,
  title     = {Stochastic Gradient Push for Distributed Deep Learning},
  author    = {Assran, Mahmoud and Loizou, Nicolas and
               Ballas, Nicolas and Rabbat, Mike},
  booktitle={International Conference on Machine Learning ({ICML})},
  volume={97},
  pages     = {344--353},
  year      = {2019},
}

@article{liu2025decentralized,
  title={Decentralized gradient tracking with local steps},
  author={Liu, Yue and Lin, Tao and Koloskova, Anastasia and Stich, Sebastian U},
  journal={Optimization Methods and Software},
  volume={40},
  number={5},
  pages={1153--1180},
  year={2025},
  publisher={Taylor \& Francis},
}

@article{alghunaim2024local,
  title={Local exact-diffusion for decentralized optimization and learning},
  author={Alghunaim, Sulaiman A},
  journal={IEEE Transactions on Automatic Control},
  volume={69},
  number={11},
  pages={7371--7386},
  year={2024},
  publisher={IEEE},
}

@article{ren2026communication,
  title={Communication-efficient stochastic distributed learning},
  author={Ren, Xiaoxing and Bastianello, Nicola and Johansson, Karl H and Parisini, Thomas},
  journal={IEEE Transactions on Automatic Control},
  year={2026},
  publisher={IEEE},
  volume={71},
  number={9},
  pages={5741--5756}
}

@article{zhang2026non,
  title={Non-convex composite federated learning with heterogeneous data},
  author={Zhang, Jiaojiao and Hu, Jiang and Johansson, Mikael},
  journal={Automatica},
  volume={183},
  pages={112695},
  year={2026},
  publisher={Elsevier},
}

@inproceedings{cheng2024momentum,
  title={Momentum benefits non-{IID} federated learning simply and provably},
  author={Cheng, Ziheng and Huang, Xinmeng and Wu, Pengfei and Yuan, Kun},
  booktitle={International Conference on Learning Representations ({ICLR})},
  volume={2024},
  pages={9815--9848},
  year={2024},
}

@article{huang2025distributed,
  title={Distributed stochastic momentum tracking with local updates: Achieving optimal communication complexities},
  author={Huang, Kun and Pu, Shi},
  journal={arXiv preprint arXiv:2510.24155},
  year={2025},
  url={https://arxiv.org/abs/2510.24155}
}

@article{huang2025beyond,
  title={Beyond {Scaffold}: A Unified Spatio-Temporal Gradient Tracking Method},
  author={Huang, Yan and Xu, Jinming and Chen, Jiming and Johansson, Karl Henrik},
  journal={arXiv preprint arXiv:2512.01732},
  year={2025},
  url={https://arxiv.org/abs/2512.01732}
}

@article{xin2020general,
  title={A general framework for decentralized optimization with first-order methods},
  author={Xin, Ran and Pu, Shi and Nedi{\'c}, Angelia and Khan, Usman A},
  journal={Proceedings of the IEEE},
  volume={108},
  number={11},
  pages={1869--1889},
  year={2020},
  publisher={IEEE},
}

@article{liang2025understanding,
  title={Understanding the influence of digraphs on decentralized optimization: Effective metrics, lower bound, and optimal algorithm},
  author={Liang, Liyuan and Huang, Xinmeng and Xin, Ran and Yuan, Kun},
  journal={SIAM Journal on Optimization},
  volume={35},
  number={3},
  pages={1570--1600},
  year={2025},
  publisher={SIAM},
}

@ARTICLE{li2025asymmetrically,
  author={Li, Qinglun and Zhang, Miao and Yin, Nan and Yin, Quanjun and Shen, Li and Cao, Xiaochun},
  journal={IEEE Transactions on Computers}, 
  title={Asymmetrically Decentralized Federated Learning}, 
  year={2025},
  volume={74},
  number={8},
  pages={2745--2756},
}

@article{chen2024accelerated,
  title={Accelerated distributed stochastic nonconvex optimization over time-varying directed networks},
  author={Chen, Yiyue and Hashemi, Abolfazl and Vikalo, Haris},
  journal={IEEE Transactions on Automatic Control},
  volume={70},
  number={4},
  pages={2196--2211},
  year={2025},
  publisher={IEEE},
}

@inproceedings{Dong2026,
  title={Accelerated push-pull methods for distributed stochastic nonconvex optimization},
  author={Dong, Sheng and Tao, Huijie and Xia, Weiguo},
  booktitle={IFAC World Congress},
  pages={},
  year={2026}
}

@article{you2026distributed,
  title={Distributed learning over arbitrary topology: Linear speed-up with polynomial transient time},
  author={You, Runze and Pu, Shi},
  journal={IEEE Transactions on Automatic Control},
  year={2026},
  publisher={IEEE}
}

@article{zhang2021fedpd,
  title={{FedPD}: A federated learning framework with adaptivity to non-{IID} data},
  author={Zhang, Xinwei and Hong, Mingyi and Dhople, Sairaj and Yin, Wotao and Liu, Yang},
  journal={IEEE Transactions on Signal Processing},
  volume={69},
  pages={6055--6070},
  year={2021},
  publisher={IEEE},
}

@inproceedings{mishchenko2022proxskip,
  title={{ProxSkip}: Yes! local gradient steps provably lead to communication acceleration! finally!},
  author={Mishchenko, Konstantin and Malinovsky, Grigory and Stich, Sebastian and Richt{\'a}rik, Peter},
  booktitle={International Conference on Machine Learning ({ICML})},
  pages={15750--15769},
  year={2022},
  volume={162},
}

@inproceedings{wang2025promise,
  title={From promise to practice: realizing high-performance decentralized training},
  author={Wang, Zesen and Zhang, Jiaojiao and Wu, Xuyang and Johansson, Mikael},
  booktitle={International Conference on Learning Representations ({ICLR})},
  volume={2025},
  pages={26204--26249},
  year={2025},
}

\end{document}